\documentclass[journal]{IEEEtran}

\usepackage{amsmath,amstext,amsfonts,amssymb,amsthm,bbold} % For maths
\usepackage{graphicx,subfigure,color} % For graphics and color

\def \A {\mathbf{A}}
\def \a {\mathbf{a}}
\def \B {\mathbf{B}}

\def \D {\mathbf{D}}
\def \d {\mathbf{d}}

\def \e {\mathbf{e}}

\def \I {\mathbf{I}}

\def \m {\mathbf{m}}

\def \R {\mathbf{R}}
\def \s {\mathbf{s}}
\def \S {\mathbf{S}}
\def \T {\mathbf{T}}

\def \U {\mathbf{U}}
\def \u {\mathbf{u}}
\def \v {\mathbf{v}}
\def \V {\mathbf{V}}
\def \W {\mathbf{W}}

\def \x {\mathbf{x}}
\def \X {\mathbf{X}}
\def \Y {\mathbf{Y}}
\def \y {\mathbf{y}}

\def \Dcal {\mathcal{D}}

\def \Ical {\mathcal{I}}

\def \Kcal {\mathcal{K}}
\def \Lcal {\mathcal{L}}

\def \Ncal {\mathcal{N}}
\def \Ocal {\mathcal{O}}

\def \Rcal {\mathcal{R}}

\def \Cbb {\mathbb{C}}
\def \Ebb {\mathbb{E}}
\def \Nbb {\mathbb{N}}
\def \Pbb {\mathbb{P}}
\def \Rbb {\mathbb{R}}
\def \Vbb {\mathbb{V}}

\def \drm {\mathrm{d}}
\def \erm {\mathrm{e}}
\def \irm {\mathrm{i}}

\def \thetabs {\boldsymbol{\theta}}

\def \Gammabs {\boldsymbol{\Gamma}}

\def \Pibs {\boldsymbol{\Pi}}
\def \Sigmabs {\boldsymbol{\Sigma}}

\def \Tr {\mathrm{tr}\,}

\def \diag{\mathrm{diag}}

\DeclareMathOperator{\sinc}{sinc}
\DeclareMathOperator*{\argmin}{argmin}
\DeclareMathOperator*{\argmax}{argmax}

\renewcommand{\Re}{\mathrm{Re}}

\newtheorem{theorem}{Theorem}

\newtheorem{lemma}{Lemma}
\newtheorem{remark}{Remark}

\newtheorem{assum}{Assumption}

\begin{document}
\title{Performance analysis of an improved MUSIC DoA estimator}

\author
{
	Pascal~Vallet,~\IEEEmembership{Member,~IEEE,}
	Xavier Mestre,~\IEEEmembership{Senior Member,~IEEE,}
    and~Philippe~Loubaton,~\IEEEmembership{Fellow,~IEEE}
	\thanks
	{
		P. Vallet is with Laboratoire de l'Int\'egration du Mat\'eriau au Syst\`eme (CNRS, Univ. Bordeaux, Bordeaux INP), 
		351, Cours de la Lib\'eration 33405 Talence (France), pascal.vallet@bordeaux-inp.fr
	}
	\thanks
	{
		X. Mestre is with Centre Tecnol\`{o}gic de Telecomunicacions de Catalunya (CTTC), 	Av. Carl Friedrich Gauss 08860 Castelldefels, Barcelona (Spain), 
		xavier.mestre@cttc.cat
	}
	\thanks
	{
		P. Loubaton is with Laboratoire d'Informatique Gaspard Monge (CNRS, Universit\'e Paris-Est/MLV), 5 Bd. Descartes 77454 Marne-la-Vall\'ee (France),
		loubaton@univ-mlv.fr
	}
    \thanks
    {
    	This work was partially supported by the Catalan Government under grant 2014 SGR 1567, the French Projects ANR-12-MONU-0003 DIONISOS and GDR ISIS 
    	Jeunes Chercheurs.
    	The material of this paper was partly presented in the conference papers \cite{mestre2011asymptotic} \cite{vallet2012clt}.
		
    }
}

\maketitle

\begin{abstract}
This paper adresses the statistical performance of subspace DoA estimation using a sensor array, in the asymptotic regime where the number of samples and sensors both converge to infinity at the same rate.  
Improved subspace DoA estimators were derived (termed as G-MUSIC) in previous works, and were shown to be consistent and asymptotically Gaussian distributed in the case where the number of sources and their DoA remain fixed. 
In this case, which models widely spaced DoA scenarios, it is proved in the present paper that the traditional MUSIC method also provides DoA consistent estimates having the same asymptotic variances as the G-MUSIC estimates. 
The case of DoA that are spaced of the order of a beamwidth, which models closely spaced sources, is also considered. 
It is shown that G-MUSIC estimates are still able to consistently separate the sources, while this is no longer the case for the MUSIC ones. 
The asymptotic variances of G-MUSIC estimates are also evaluated.   
\end{abstract}


\section{Introduction}
\label{section:introduction}

\IEEEPARstart{T}{he} problem of estimating the directions of arrival (DoA) of source signals with an array of sensors is fundamental in statistical signal processing, 
and several methods have been developed and characterized in terms of performance, during the past 40 years.
Among the most popular high resolution methods, subspace algorithms such as MUSIC \cite{schmidt1986multiple} are widely used. 
It is well known (see e.g. \cite{thomas1995probability}) that subspace methods suffer the so-called ``threshold effect", which involves a severe degradation when either 
the Signal to Noise Ratio (SNR) and/or the sample size are not large enough. In contrast, the threshold breakdown is less significant for Maximum Likelihood (ML) techniques, and occurs for a much lower SNR and/or sample size. However, due to their reduced complexity since they involve a one-dimensional search over the set of possible DoA, subspace methods are usually prefered over ML which requires a multi-dimensional search.

The study of the statistical performance of MUSIC algorithm has received a lot of attention, see e.g. \cite{stoica1989music}, and its behaviour has been mainly 
characterized in the situation where the number of available samples $N$ of the observed signal is much larger than the number of sensors $M$ of the array. However, there may exist some situations where this hypothesis is not realistic, for example when the number of sensors $M$ is large and the 
signals have short-time duration or short time stationarity. In this case, $M$ and $N$ are of the same order of magnitude, 
and the standard statistical analysis of MUSIC is irrelevant. This is mainly because the sample correlation matrix of the observations, on which MUSIC mainly relies, does not properly estimate the true covariance matrix. In this 
context, the standard estimate of the MUSIC angular ``pseudo-spectrum" does not appear to be consistent. 
To model this more stringent scenario, it was proposed in \cite{mestre2008modified} to consider a new asymptotic regime in which both $M,N$ converges to infinity at the same rate, 
that is 
\begin{center}
	$M,N \to \infty$ such that $\frac{M}{N} \to c > 0$.
\end{center}
Based on results from random matrix theory, giving a precise description of the behaviour of the eigenvalues and eigenvectors of large random matrices, 
an improved MUSIC DoA technique, termed as ``G-MUSIC", was derived in \cite{mestre2008modified} in the unconditional model case, that is, 
by assuming that the source signals are Gaussian and temporally white. 
This method was based on a novel estimator of the ``pseudo-spectrum" function. 
Other related works concerning the unconditional case include \cite{abramovich-mestre-2008} as well as 
\cite{krichtman-nadler-2009} where the source number detection is addressed. 
Later, \cite{vallet2012improved} addressed the more general conditional model case, i.e. the source signals 
are modelled as non observable deterministic signals. Using an approach similar to \cite{mestre2008modified}, a different estimator of the pseudo-spectrum was proposed. 
More recently, the work of \cite{vinogradova2013statistical} extends the improved subspace estimation of \cite{vallet2012improved} to the situation where the noise may be correlated in time. We also mention the recent series of works \cite{couillet2014robust} \cite{couillet2015robust} \cite{couillet2015random} on robust subspace estimation, in the context of impulsive noise.

Experimentally, it can be observed that in certain scenarios, 
MUSIC and G-MUSIC show quite similar performance, while in other contexts G-MUSIC outperforms MUSIC. In this paper which is focused on the conditional case, we explain this behaviour and provide a complete description 
of the statistical performance of MUSIC and G-MUSIC. Roughly speaking, we prove that 
if the DoAs are widely spaced compared to $\frac{1}{M}$, MUSIC and G-MUSIC have a similar behaviour, while MUSIC fails when the DoAs are closely spaced. More precisely, we establish the following results.   
\begin{itemize}
	\item When the number of sources $K$ and the corresponding DoA remain fixed as $M,N \to \infty$ (a regime which models widely spaced sources), we show that,
	while the pseudo-spectrum estimate of MUSIC is inconsistent, its minimization w.r.t. the DoA provides $N$-consistent
\footnote
{
	An estimator $\hat{\theta}_N$ of a (possibly depending on $N,M$) DoA $\theta_N$ is
	defined as $N$-consistent if almost surely, $N\left(\hat{\theta}_N-\theta_N\right) \to 0$ as $M,N \to \infty$.
}
 estimates. Moreover, in the case of asymptotically uncorrelated source signals, the MUSIC
 DoA estimates share the same asymptotic MSE as G-MUSIC. 
%	Moreover, we show that a basic spatial periodogram also generate consistent DoA estimators in this context, which confirms the well-known fact that low resolutions
%	techniques based on spatial filtering have performance as good as subspace methods for widely spaced DoA and low sample size situations.
	\item For two sources with an angular spacing of the order of a beamwidth, that is $\Ocal(M^{-1})$ as $M,N \to \infty$, we show that G-MUSIC remains $N$-consistent while MUSIC is not $N$-consistent anymore, which means that MUSIC is no longer able to asymptotically separate the DoA. 
\end{itemize}

		\subsection{Problem formulation and previous works}

Let us consider the situation where $K$ narrow-band and far-field source signals are impinging on a uniform linear array of $M$ sensors, with $K < M$.
The received signal at the output of the array is usually modeled as a complex $M$-variate time series $(\y_n)_{n \geq 1}$ given by
\begin{align}
	\y_n = \A \s_n + \v_n,
	\notag
\end{align}
where 
\begin{itemize}
	\item $\A = [\a(\theta_1),\ldots,\a(\theta_K)]$ is the $M \times K$ matrix of steering vectors $\a(\theta_1),\ldots,\a(\theta_K)$, with 
	$\theta_1,\ldots,\theta_K$ the source signals DoA, and $\a(\theta)= \frac{1}{\sqrt{M}}[1,\ldots,\erm^{\irm (M-1) \theta}]^T$ ;
	\item $\s_n \in \Cbb^K$ contains the source signals received at time $n$, considered as unknown deterministic ; 
	\item $(\v_n)_{n \geq 1}$ is a temporally and spatially white circularly symmetric complex Gaussian noise with spatial covariance $\Ebb[\v_n\v_n^*]=\sigma^2 \I$.
\end{itemize}
By assuming that $N$ observations $\y_1,\ldots,\y_N$ are collected in the $M \times N$ matrix
\begin{align}
	\Y_N = [\y_1,\ldots,\y_N] = \A\S_N + \V_N,
	\label{eq:model_signal_N}
\end{align}
with $\S_N=[\s_1,\ldots,\s_N]$ and $\V_N = [\v_1,\ldots,\v_N]$, the DoA estimation problem thus consists in estimating the $K$ DoA $\theta_1,\ldots,\theta_K$ 
from the matrix of samples $\Y_N$. 

Subspace methods are based on the observation that the source contributions $\A \s_1,\ldots,\A\s_N$ are confined in the so-called signal subspace of dimension $K$, defined as 
$\mathrm{span}\left\{\a(\theta_1),\ldots,\a(\theta_K)\right\}$ .
By assuming that the signal sample covariance $N^{-1} \S_N \S_N^*$ is full rank, $\theta_1,\ldots,\theta_K$ are the unique zeros of the pseudo-spectrum
\begin{align}
	\eta(\theta) = \a(\theta)^* \Pibs \a(\theta),
\end{align}
where $\Pibs$ is the orthogonal projection matrix onto the noise subspace, defined as the orthogonal complement of the signal subspace, 
and which coincides in that case with the kernel of $N^{-1} \A \S_N\S_N^* \A^*$ of dimension $M-K$.

Since $\Pibs$ is not available in practice, it must be estimated from the observation matrix $\Y_N$.
This estimation is traditionnaly performed by using the so-called sample correlation matrix of the observations (SCM)
\begin{align}
	\frac{\Y_N\Y_N^*}{N} = \frac{1}{N} \sum_{n=1}^N \y_n\y_n^*,
	\notag
\end{align}
and $\Pibs$ is directly estimated by considering its sample estimate $\hat{\Pibs}_N$, i.e. the corresponding orthogonal projection matrix onto the eigenspace 
associated with the $M-K$ smallest eigenvalues of $\frac{\Y_N\Y_N^*}{N}$. 
The MUSIC method thus consists in estimating the DoA $\theta_1,\ldots,\theta_K$ as the $K$ most significant minima of the estimated pseudo-spectrum
\begin{align}
	\hat{\eta}_N^{(t)}(\theta) = \a(\theta)^* \hat{\Pibs}_N \a(\theta),
	\notag
\end{align}
where the superscript $^{(t)}$ refers to ``traditional estimate".

The SCM is known to be an accurate estimator of the true covariance matrix when the number of available samples $N$ is much larger than the observation 
dimension $M$. Indeed, in the asymptotic regime where $M$ is constant and $N$ converges to infinity, under some technical conditions, 
the law of large numbers ensures that
\begin{align}
	\left\|\frac{\Y_N\Y_N^*}{N} - \left(\A\frac{\S_N\S_N^*}{N}\A^* + \sigma^2 \I\right)\right\| \to 0,
%	\notag
\label{eq:consistency-YY*}
\end{align}
almost surely (a.s.) as $N \to \infty$, where $\|.\|$ stands for the spectral norm. This implies that
\begin{align}
	\left\|\hat{\Pibs}_N - \Pibs\right\| \xrightarrow[N \to \infty]{a.s.} 0 
	\label{eq:scm_consistency_trad}
\end{align}
i.e. the sample projection matrix $\hat{\Pibs}_N$ is a consistent estimator of $\Pibs$. 
%This sample correlation matrix has moreover been characterized in terms of asymptotic distribution (cf Anderson \cite{anderson1958introduction}) and
%mean square error (MSE), in the previous asymptotic regime. 
Moreover, \eqref{eq:scm_consistency_trad} directly implies the uniform consistency of the traditional pseudo-spectrum estimate
\begin{align}
	\sup_{\theta \in [-\pi,\pi]} \left|\hat{\eta}_N^{(t)}(\theta)-\eta(\theta)\right| \xrightarrow[N \to \infty]{a.s.} 0.
%	\notag
\label{eq:consistency-trad-pseudo-spectrum}
\end{align}
The $K$ MUSIC DoA estimates, defined formally, for $k=1,\ldots,K$, by 
\begin{align}
	\hat{\theta}^{(t)}_{k,N} = \argmin_{\theta \in \Ical_k} \hat{\eta}_N^{(t)}(\theta),
	\notag
\end{align}
where $\Ical_k$ is a compact interval containing $\theta_k$ and such that $\Ical_k \cap \Ical_l = \emptyset$ for $k \neq l$,
are therefore consistent, i.e. 
\begin{align}
	\hat{\theta}_{k,N}^{(t)} \xrightarrow[N \to \infty]{a.s.} \theta_k.
	\notag
\end{align}
Several accurate approximations of the MSE on the MUSIC DoA estimates have been obtained (see e.g. \cite{stoica1989music} and the references therein).

In the situation where $M,N$ are of the same order of magnitude, \eqref{eq:consistency-YY*}, 
and therefore \eqref{eq:scm_consistency_trad} as well as \eqref{eq:consistency-trad-pseudo-spectrum}, are no longer true.
To analyze this situation, \cite{mestre2008modified} proposed to consider the non standard asymptotic regime in which 
\begin{align}
	M,N \to \infty \text{ such that } \frac{M}{N} \to c > 0.
	\label{eq:asymptotic_regime}
\end{align}

In \cite{vallet2012improved}, an estimator $\hat{\eta}_N(\theta)$ of the pseudo-spectrum $\eta(\theta)$ was derived. Under an extra assumption, called the separation condition, it was proved to be consistent in the new asymptotic regime \eqref{eq:asymptotic_regime}, that is
\begin{align}
	\hat{\eta}_N(\theta) - \eta(\theta) \xrightarrow[]{} 0,
	\notag
\end{align}
almost surely, when
\footnote
{
	Note that in that case $\eta(\theta)$ depends on $M$ (and thus implicitely on $N$). 
	In the next sections, a subscript $N$ will be added to make clear this dependence.
}
$M,N \to \infty$ such that $\frac{M}{N} \to c > 0$. In the case where the number of sources $K$ remains 
fixed when $M$ and $N$ increase, the separation condition was shown to hold if the eigenvalues of $\A\frac{\S_N\S_N^*}{N}\A^*$ are above the threshold 
$\sigma^{2} \sqrt{c}$ \cite[Section III-C]{vallet2012improved}.  
Note that a similar estimator was previously derived in \cite{mestre2008modified} in the unconditional source signal case.
A stronger result of uniform convergence over $\theta$ was proved in \cite{hachem2012large}, that is
\begin{align}
	\sup_{\theta \in [-\pi,\pi]} \left|\hat{\eta}_N(\theta)-\eta(\theta)\right| \to 0,
	\notag
\end{align}
almost surely. 
When $K$ and the DoA $(\theta_k)_{k=1, \ldots, K}$ remain fixed, the G-MUSIC DoA estimates, defined for $k=1,\ldots,K$ by $\hat{\theta}_{k,N}= \argmin_{\theta \in \Ical_k}\hat{\eta}_N(\theta)$, were also shown to be $N$-consistent, that is
\begin{align}
	N \left(\hat{\theta}_{k,N} - \theta_k\right) \to 0
	\notag
\end{align}
almost surely, when $M,N \to \infty$ such that $\frac{M}{N} \to c > 0$. 
More recently, \cite{hachem2012subspace} also proposed a second-order analysis of the G-MUSIC DoA estimates (in the conditional case), 
in terms of a Central Limit Theorem (CLT) in the latter asymptotic regime.

The work in \cite{hachem2012subspace} assumes that the source signals are spatially uncorrelated asymptotically, that is $N^{-1}\S_N\S_N^*$ converges to a positive 
diagonal matrix as $N \to \infty$, and both \cite{hachem2012large} and \cite{hachem2012subspace} that the source DoA $\theta_1,\ldots,\theta_K$ are fixed with respect to $M,N$. 
This latter assumption is suitable for practical scenarios in which the source DoA are widely spaced. 
However, for scenarios in which the source DoA are closely spaced, e.g. with an angular separation of the order $\Ocal\left(M^{-1}\right)$), 
the analysis of G-MUSIC provided in \cite{hachem2012large} and \cite{hachem2012subspace} are not relevant anymore.

In this paper, we address a theoretical comparison between the performance of MUSIC and G-MUSIC in the two following scenarios.

In a first scenario, in which the number of sources $K$ and the corresponding DoA $\theta_1,\ldots,\theta_K$ are considered fixed with respect to $M,N$ (referred to as ``widely spaced DoA") and where it is known that G-MUSIC is $N$-consistent, we prove that, while the traditional MUSIC pseudo-spectrum estimate $\hat{\eta}^{(t)}_N(\theta)$ is inconsistent, the  MUSIC algorithm is $N$-consistent and that the two methods exhibit the same asymptotic Gaussian distributions.
We remark that the analysis provided for this scenario allows spatial correlation between the different source signals. 
%We also show in this scenario that a simple averaged periodogram estimates consistently the source DoA.

In a second scenario, we consider $K=2$ spatially uncorrelated source signals with DoA $\theta_{1}$ and $\theta_{2}$ depending on $M,N$ such that their angular separation $\theta_1 - \theta_2 = \Ocal\left(M^{-1}\right)$, when $M,N$ converge to infinity at the rate. 
We show in this context that the G-MUSIC DoA estimates remain $N$-consistent while MUSIC looses its $N$-consistency. 
We also provide in this scenario the asymptotic distribution for the G-MUSIC DoA estimates.

To obtain the asymptotic distribution of G-MUSIC under the two previous scenarios, we rely on a Central Limit Theorem (CLT) which extends the results obtained in \cite{hachem2012subspace} using a different approach, and which allows situations involving spatial correlations between sources and closely spaced DoA. A CLT for the traditional MUSIC DoA estimates is also given in the first scenario using the same technique.
The proofs of these results need the use of large random matrix theory technics, and appear to be quite long and technical. Therefore, we choose to not include them in the present paper. However, the derivations are available on-line at \cite{arxiv_clt}.

		\subsection{Organization and notations}

\textit{Organization of the paper:}
In section \ref{section:review_rmt}, we review some basic random matrix theory results, concerning the asymptotic behaviour of the eigenvalues of the SCM in the case where the number of sources $K$ remains fixed when $M$ and $N$ increase. We then make use of these results to introduce the estimator of any bilinear form of the noise subspace projector $\Pibs$, derived in \cite{vallet2012improved}. 
We also give a Central Limit Theorem (CLT) for this estimator, which will be used in the subsequent sections to derive the asymptotic distribution of the G-MUSIC DoA estimates.
In section \ref{section:consistency}, we prove that MUSIC and G-MUSIC are both $N$-consistent in the scenario where the source DoA are widely spaced. 
However, in a closely spaced DoA scenario, we prove that MUSIC is not $N$-consistent, while G-MUSIC is still $N$-consistent.
Finally, we provide in section \ref{section:clt} an analysis of G-MUSIC and MUSIC DoA estimates in terms of Asymptotic Gaussianity. In particular, it is shown that MUSIC and G-MUSIC exhibit exactly the same asymptotic MSE in the widely spaced DoA scenario 
and for asymptotically uncorrelated source signals.
Some numerical experiments are provided which confirm the accuracy of the predicted performance of both methods.

\textit{Notations:}
For a complex matrix $\A$, we denote by $\A^T, \A^*$ its transpose and its conjugate transpose, and by $\Tr(\A)$ and $\|\A\|$ its trace and spectral norm. 
The identity matrix will be $\I$ and $\e_n$ will refer to a vector having all its components equal to $0$ except the $n$-th equals to $1$.
The notation $\mathrm{span}\{\x_1,\ldots,\x_n\}$ will refer to the vector space generated by $\x_1,\ldots,\x_n$.
The real normal distribution with mean $m$ and variance $\sigma^2$ is denoted $\Ncal_{\Rbb}(\alpha,\sigma^2)$ and the multivariate normal distribution in $\Rbb^k$, with mean $\m$ and covariance $\Gammabs$ is denoted in the same way $\Ncal_{\Rbb^k}(\m,\Gammabs)$. 
A complex random variable $Z = X +\irm Y$ follows the distribution $\Ncal_{\Cbb}(\alpha+\irm\beta,\sigma^2)$ if $X$ and $Y$ are independent with respective distributions 
$\Ncal_{\Rbb}(\alpha, \frac{\sigma^2}{2})$ and $\Ncal_{\Rbb}(\beta, \frac{\sigma^2}{2})$. 
The expectation and variance of a complex random variable $Z$ will be denoted $\Ebb[Z]$ and $\Vbb[Z]$.
For a sequence of random variables $(X_n)_{n \in \Nbb}$ and a random variable $X$, we write
\begin{align}
	X_n \xrightarrow[n\to\infty]{a.s.} X \text{ and } X_n \xrightarrow[n\to\infty]{\Dcal} X
	\notag
\end{align}
when $X_n$ converges respectively with probability one and in distribution to $X$. Finally, $X_n = o_{\Pbb}(1)$ will stand for the convergence of $X_n$ to $0$ in probability, and
$X_n = \Ocal_{\Pbb}(1)$ will stand for tightness (boundedness in probability).

\section{Asymptotic behaviour of the sample eigenvalues and eigenvectors}
\label{section:review_rmt}

In this section, we present some basic results from random matrix theory describing the behaviour of the eigenvalues of the SCM $\frac{\Y_N\Y_N^*}{N}$, in the asymptotic regime where $M,N$ converge to infinity such that $\frac{M}{N} \to c > 0$. 
These results are required to properly introduce the improved subspace estimator of \cite{vallet2012improved}.
To that end, we will work with the following more general model, referred to as ``Information plus Noise" in the literature.

We consider $M,N,K \in \Nbb^*$ such that $K < M$ and $M=M(N)$, is a function of $N$ satisfying 
\footnote
{
	The condition $\sqrt{N}\left(c_N-c\right) \to 0$ is purely technical and is in fact only needed for 
	the validity of Theorems \ref{theorem:clt_subspace}, \ref{theorem:clt_MUSIC} and \ref{theorem:clt_GMUSIC} below.
	}
\begin{equation}
	\label{eq:condition-cN}
	c_N = \frac{M}{N} = c + o\left(\frac{1}{\sqrt{N}}\right)
\end{equation}	
as $N \to \infty$. 
Thus, in the remainder, the notation $N \to \infty$ will refer to the double asymptotic regime $M,N \to \infty$, $M/N \to c > 0$.
We also assume that $K$ is fixed with respect to $N$ (for the general case where $K$ may possibly go to infinity with $N$, see \cite{vallet2012improved}).
We consider the sequence of random matrices $\left(\Sigmabs_N\right)_{N \geq 1}$ of size $M \times N$ where
\footnote
{
	Of course, we retrieve the usual array processing model \eqref{eq:model_signal_N} by
	setting $\Sigmabs_N=N^{-1/2}\Y_N$, $\B_N=N^{-1/2}\A\S_N$ and $\W_N = N^{-1/2}\V_N$.
}
\begin{align}
	\Sigmabs_N = \B_N + \W_N,
	\label{eq:INmodel}
\end{align}
with
\begin{itemize}
 \item $\B_N$ a rank $K$ deterministic matrix satisfying $\sup_{N} \|\B_N\| < \infty$,
 \item $\W_N$ having i.i.d. $\Ncal_{\Cbb}\left(0,\frac{\sigma^2}{N}\right)$ entries .
\end{itemize}
We denote by $\lambda_{1,N} \geq \ldots \geq \lambda_{K,N}$ the non zero eigenvalues of $\B_N\B_N^*$ and by $\u_{1,N},\ldots,\u_{K,N}$ the respective unit norm eigenvectors. $(\u_{k,N})_{k=K+1, \ldots, M}$ are unit norm mutually 
orthogonal vectors of the kernel of $\B_N\B_N^*$. 
Equivalently, $\hat{\lambda}_{1,N}\geq\ldots\geq\hat{\lambda}_{M,N}$ are the eigenvalues of the matrix $\Sigmabs_N\Sigmabs_N^*$ and $\hat{\u}_{1,N},\ldots,\hat{\u}_{M,N}$ the respective unit norm eigenvectors.

	\subsection{The asymptotic spectral distribution of the SCM}
		
Let $\hat{\mu}_N$ be the empirical spectral measure of the matrix $\Sigmabs_N\Sigmabs_N^*$, defined as the random probability measure
\begin{align}
	\hat{\mu}_N = \frac{1}{M} \sum_{k=1}^M \delta_{\hat{\lambda}_{k,N}},
	\notag
\end{align}	
with $\delta_x$ the Dirac measure at point $x$.
The distribution $\hat{\mu}_N$ can be alternatively characterized through its Stieltjes transform defined as
\begin{align}
	\hat{m}_N(z) = \int_{\Rbb} \frac{\drm \hat{\mu}_N(\lambda)}{\lambda - z} = \frac{1}{M} \Tr \left(\Sigmabs_N\Sigmabs_N^* - z \I\right)^{-1}
	\notag
\end{align}
where $\left(\Sigmabs_N\Sigmabs_N^* - z \I\right)^{-1}$ is the resolvent of the matrix $\Sigmabs_N\Sigmabs_N^*$.

It is well-known from \cite{marchenko1967distribution} that for all $z \in \Cbb \backslash \Rbb$,
\begin{align}
	\hat{m}_N(z) \xrightarrow[N \to \infty]{a.s.} m(z),
	\label{eq:conv_stieltjes}
\end{align}
where 
\begin{align}
	m(z) = \int_{\Rbb} \frac{\drm \mu(\lambda)}{\lambda - z}
	\notag
\end{align}
is the Stieltjes of a deterministic probability measure called the Marchenko-Pastur distribution, whose support coincides with the compact interval 
$[\sigma^2 (1-\sqrt{c})^2, \sigma^2 (1+\sqrt{c})^2]$, and which is defined by
\begin{align}
	&\drm \mu(x) = \notag\\
	&\qquad\left(1 - \frac{1}{c}\right)^{+} \delta_{0}
	+
	\frac{\sqrt{\left(x - x^-\right)\left(x^+ - x\right)}}{2 \sigma^2 c \pi x} \mathbb{1}_{[x^-,x^+]}(x) \drm x.
	\notag
\end{align}
with $x^- = \sigma^2 (1-\sqrt{c})^2$ and $x^+ = \sigma^2 (1+\sqrt{c})^2$.

Moreover, $m(z)$ satisfies the following fundamental equation 
\begin{align}
	m(z) = \frac{1}{-z\left(1+\sigma^2 c m(z)\right) + \sigma^2 (1-c)}.
	\label{eq:fundamental_equation_m}
\end{align}
An equivalent statement of \eqref{eq:conv_stieltjes} is given with the following convergence  in distribution
\begin{align}
	\hat{\mu}_N \xrightarrow[N \to \infty]{\Dcal} \mu
	\notag
\end{align}
which holds almost surely, that is, the empirical eigenvalue distribution of $\Sigmabs_N\Sigmabs_N^*$ has the same asymptotic behaviour as the Marchenko-Pastur distribution.
Practically, the eigenvalue histogram of $\Sigmabs_N\Sigmabs_N^*$ matches the density of the Marchenko-Pastur distribution, for $M,N$ large enough, as shown in Figure \ref{fig:MP}, where
we have chosen $M=1000$, $N=2000$, $\sigma^2=1$ and $K=2$ with $\lambda_{1,N}=5$ and $\lambda_{2,N}=10$.

\begin{figure}[!h]
	\centering
	\includegraphics[scale=0.4]{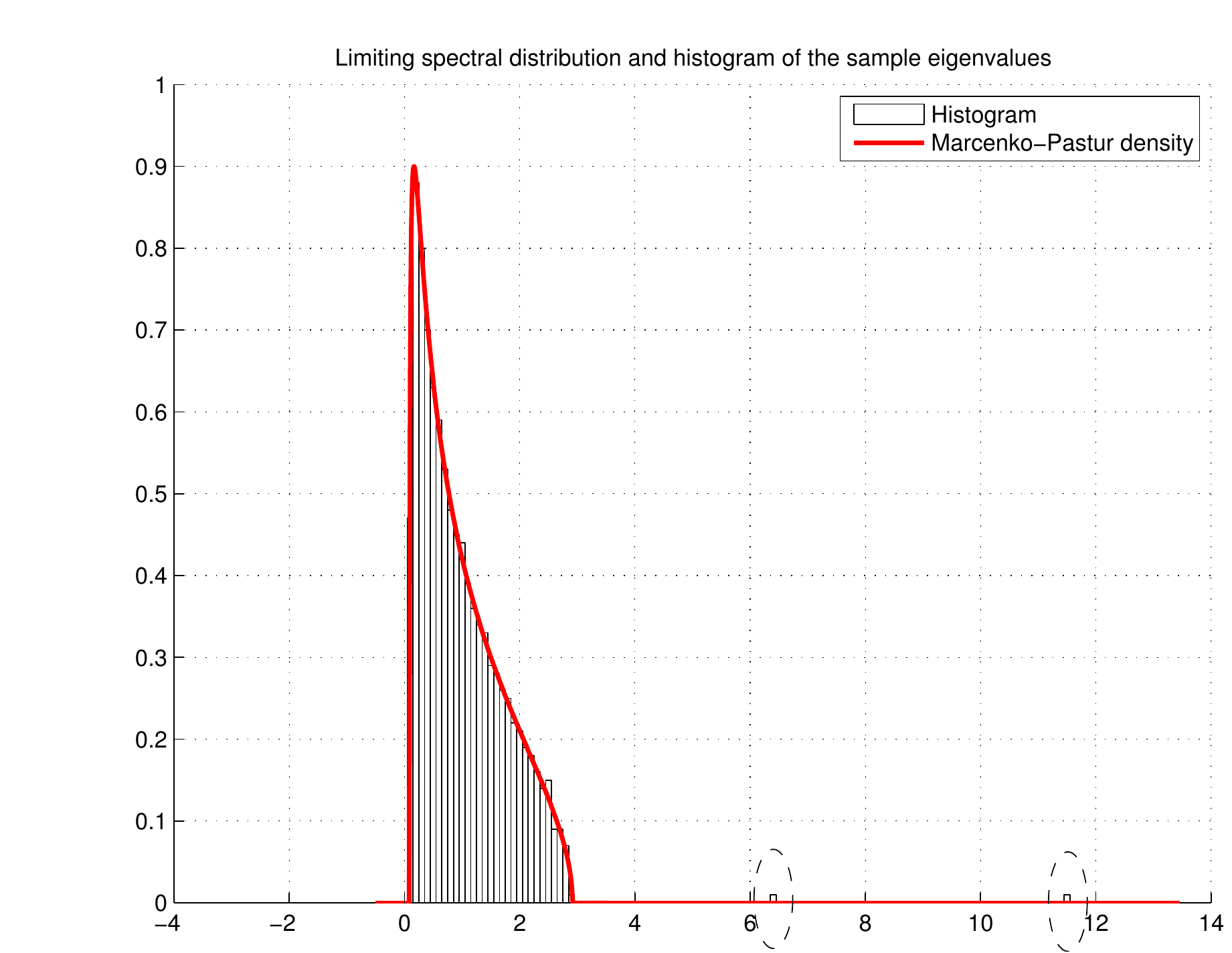}
	\caption{Marchenko-Pastur distribution and eigenvalue histogram of $\Sigmabs_N\Sigmabs_N^*$}
	\label{fig:MP}
\end{figure}

\begin{remark}
	The Marchenko-Pastur distribution was originally obtained as the limit distribution of the empirical eigenvalue distribution of the noise part $\W_N\W_N^*$.
	Nevertheless, the assumption that the rank $K$ of the deterministic perturbation $\B_N$ is independent of $N$ implies that the Marchenko-Pastur limit still holds for 
	$\Sigmabs_N\Sigmabs_N^*$. This fact is well known, and can be easily seen by expressing $\hat{m}_N(z)$ in terms of the Stieltjes transform of the spectral distribution of $\W_N\W_N^*$.
	Finite rank perturbations of $\W_N$ are often referred to as ``spiked models" in the random matrix literature \cite{benaych2012singular}.
\end{remark}

	\subsection{Asymptotic behaviour of the sample eigenvalues}
	
As also noticed in Figure \ref{fig:MP}, the non zero eigenvalues $\lambda_{1,N}$ and $\lambda_{2,N}$ of $\B_N\B_N^*$ generate two outliers $\hat{\lambda}_{1,N}$, $\hat{\lambda}_{2,N}$
in the spectrum of $\Sigmabs_N\Sigmabs_N^*$, in the sense that $\hat{\lambda}_{1,N}$, $\hat{\lambda}_{2,N}$ are outside the support $[x^-,x^+]$ of the Marchenko-Pastur distributions,
while all the remaining eigenvalues $\hat{\lambda}_{3,N},\ldots,\hat{\lambda}_{M,N}$ concentrate around $[x^-,x^+]$.

In fact, under an additional condition on the non zero eigenvalues $\lambda_{1,N},\ldots,\lambda_{K,N}$, it is possible to characterize the behaviour of the $K$ largest sample eigenvalues $\hat{\lambda}_{1,N},\ldots,\hat{\lambda}_{K,N}$.
The following assumption, usually referred to as \textit{subspace separation condition}, ensures that the $K$ non zero eigenvalues of $\B_N\B_N^*$ are sufficiently separated from 
the $M-K$ zero eigenvalues.

\begin{assum}
	\label{assumption:spiked}
	For $k=1,\ldots,K$, $\lambda_{k,N} \to \lambda_k$ as $N \to \infty$, where 
	\begin{align}
		\lambda_{1} > \ldots > \lambda_{K} > \sigma^2 \sqrt{c}.
		\notag
	\end{align}
\end{assum}

We note that forthcoming results remain valid if some $(\lambda_{k})_{k=1, \ldots, K}$ 
coincide. We assume that $\lambda_{k} \neq \lambda_{l}$ for $k \neq l$ in order to simplify the presentation. 
Under the previous assumption, an accurate description of the behaviour of the eigenvalues of $\Sigmabs_N\Sigmabs_N^*$ can be obtained.

\begin{theorem}
	\label{theorem:spiked_eig}
	Under Assumption \ref{assumption:spiked}, for $k=1,\ldots,K$,
	\begin{align}
		\hat{\lambda}_{k,N} \xrightarrow[N\to\infty]{a.s.} \phi(\lambda_k) = \frac{(\lambda_k + \sigma^2)(\lambda_k + \sigma^2 c)}{\lambda_k}.
		\notag
	\end{align}
	with $\phi(\lambda_k) > x^+$. Moreover, for all $\epsilon > 0$,
	\begin{align}
		\hat{\lambda}_{K+1,N},\ldots,\hat{\lambda}_{M,N} \in \left(x^- - \epsilon, x^+ + \epsilon\right),
		\notag
	\end{align}
	almost surely for $N$ large enough.
\end{theorem}
Theorem \ref{theorem:spiked_eig} is a consequence of the general results proved in \cite{benaych2012singular} (see also \cite{loubaton2011almost} for a different, but less general, proof).
Rephrased in another way, under the separation condition, the $K$ largest eigenvalues of $\Sigmabs_N\Sigmabs_N^*$ escape from the support of the Marchenko-Pastur distribution while the smallest $M-K$ eigenvalues are concentrated in a neighborhood of $[x^-,x^+]$.

\begin{remark}
	Theorem \ref{theorem:spiked_eig} in conjunction with Assumption \ref{assumption:spiked} have a nice interpretation (see e.g. \cite{nadakuditi-edelman-2008} and \cite{bianchi-et-al-2011} in the conditional case). 
	Indeed, we notice that the separation condition can be interpreted as a detectability threshold on the SNR condition, if we define the SNR to be the ratio 
	$\frac{\lambda_K}{\sigma^2}$.
	Therefore, Theorem \ref{theorem:spiked_eig} ensures that the $K$ ``signal sample eigenvalues" $\hat{\lambda}_{1,N},\ldots,\hat{\lambda}_{K,N}$ will be detectable in the sense that
	they will split from the $M-K$ ``noise sample eigenvalues" as $N \to \infty$, as long as the SNR is above $\sqrt{c}$.
	%This in fact implies a separation between noise and signal subspaces of $\Sigmabs_N\Sigmabs_N^*$.
\end{remark}

	\subsection{Estimation of the signal subspace}

In this section, we introduce a consistent estimator of any bilinear form of the noise subspace orthogonal projection matrix, which was derived in \cite{hachem2012subspace} (see also \cite{vallet2012improved}).
%In fact, it will be more convenient to work with the signal subspace orthogonal projection matrix
%\begin{align}
%	\Pibs_N^{\perp} = \I - \Pibs_N = \sum_{k=1}^K \u_{k,N}\u_{k,N}^*.
%	\notag
%\end{align}
Let us introduce the function 
\begin{align}
	w(z) = z \left(1+\sigma^2 c m(z)\right)^2 - \sigma^2 (1-c)\left(1+\sigma^2 c m(z)\right).
	\notag
\end{align}
From the fixed point equation \eqref{eq:fundamental_equation_m}, straightforward algebra leads to the new equation
\begin{align}
	\phi\left(w(z)\right) = z,
\end{align}
and one can see easily that the function $\lambda \mapsto \phi(\lambda)$ is a one to one correspondence from $\left(\sigma^2 \sqrt{c},+\infty\right)$ onto $\left(\sigma^2 (1+\sqrt{c})^2,+\infty\right)$ with inverse function $x \mapsto w(x)$ defined on the
interval $\left(\sigma^2 (1+\sqrt{c})^2,+\infty\right)$ (see \cite{arxiv_clt}).

The following fundamental result was proved in \cite{benaych2012singular} (see also \cite{hachem2012subspace}).
\begin{theorem}
	\label{theorem:spiked_eigv}
	Under Assumption \ref{assumption:spiked}, for all deterministic sequences of unit norm vectors $(\d_{1,N})$, $(\d_{2,N})$, we have for $k=1,\ldots,K$
	\begin{align}
		&\d_{1,N}^* \hat{\u}_{k,N} \hat{\u}_{k,N}^* \d_{2,N} = 
		\notag \\
		&\qquad\qquad h\left(\phi(\lambda_k)\right) \d_{1,N}^* \u_{k,N} \u_{k,N}^* \d_{2,N} + o(1) \quad a.s.,
		\notag
	\end{align}
	where 
	\begin{align}
		h(z)= \frac{w(z)^2 - \sigma^4 c}{w(z)\left(w(z)+\sigma^2 c\right)}.
		\notag
	\end{align}
\end{theorem}
Since the function $\phi$ is the inverse of the function $w$, we obtain an explicit expression for $h\left(\phi(\lambda_k)\right)$ :
\begin{align}
	h\left(\phi(\lambda_k)\right) = \frac{\lambda_k^2 - \sigma^4 c}{\lambda_k\left(\lambda_k+\sigma^2 c\right)}.
	\notag
\end{align}
Define the following bilinear form of the noise subspace orthogonal projection matrix:
\begin{align}
	\eta_N = \d_{1,N}^* \Pibs_N \d_{2,N},
	\label{eq:eta}
\end{align}
as well as its traditional estimate
\begin{align}
	\hat{\eta}_N^{(t)} = \d_{1,N}^* \hat{\Pibs}_N \d_{2,N}.
	\label{eq:eta_hat_t}
\end{align}
Then Theorem \ref{theorem:spiked_eigv} shows in particular that
\begin{align}
	\hat{\eta}_N^{(t)} = \d_{1,N}^* \left(\I - \sum_{k=1}^K h\left(\phi(\lambda_k)\right) \u_{k,N}\u_{k,N}^*\right) \d_{2,N} + o(1),
\end{align}
a.s., which implies that the traditional subspace estimate is not consistent.

Moreover, Theorem \ref{theorem:spiked_eig} in conjunction with Theorem \ref{theorem:spiked_eigv} directly provides a consistent estimator of \eqref{eq:eta}.
Indeed, under Assumption \ref{assumption:spiked},
\begin{align}
	\hat{\eta}_N - \eta_N \xrightarrow[N\to\infty]{a.s.} 0,
	\label{eq:consistency_subspace_estimator}
\end{align}
where
\begin{align}
	\hat{\eta}_N = \d_{1,N}^* \left(\I - \sum_{k=1}^K \frac{1}{h\left(\hat{\lambda}_{k,N} \right)} \hat{\u}_{k,N} \hat{\u}_{k,N}^*\right) \d_{2,N}.
	\label{eq:noise_subspace_estimator}
\end{align}

\begin{remark}
	It should be noticed that the estimator given in \eqref{eq:noise_subspace_estimator} provides in particular a consistent estimator of any $(i,j)-th$ 
	entry of $\Pibs_N$, by choosing $\d_{1,N}=\e_i$ and $\d_{2,N}=\e_j$.
	However, \eqref{eq:consistency_subspace_estimator} does not imply that we have a norm-consistent estimator of $\Pibs_N$, in the sense that 
	\begin{align}
		\left\|\Pibs_N - \left(\I - \sum_{k=1}^K \frac{1}{h\left(\hat{\lambda}_{k,N} \right)} \hat{\u}_{k,N} \hat{\u}_{k,N}^*\right)\right\|
		\notag
	\end{align}
	does not necessarily converge to $0$ as $N \to \infty$. 
	%For such an estimator in the special context of DoA estimation, we refer the reader to \cite{vallet2012toeplitz}.
\end{remark}

% 
% EQUATION DOUBLE COLONNE
%

\begin{figure*}[!t]
\normalsize
\begin{align}
\chi_{k,\ell}^{(t)}
	&= 
	\lambda_k \lambda_{\ell} \left(\lambda_k \lambda_{\ell} + \sigma^2(\lambda_k + \lambda_{\ell}) + \sigma^4\right)
	\left((1+c) (\lambda_k\lambda_{\ell}+\sigma^4 c) + 2 \sigma^2 c (\lambda_k + \lambda_{\ell})\right)
	\notag\\
	&\qquad%\qquad\qquad\qquad 
	- c \left(\lambda_k\lambda_{\ell}-\sigma^4 c\right)\left(\lambda_k\lambda_{\ell} + \sigma^2 (\lambda_k+\lambda_{\ell}) + \sigma^4 c\right)^2.
	\notag
\end{align}
\hrulefill
\vspace*{4pt}
\end{figure*}

% 
% FIN EQUATION DOUBLE COLONNE
%

A result concerning the asymptotic Gaussianity of the estimator $\hat{\eta}_N$ can be also derived. 
% To that end, we introduce the quantity
% \begin{align}
% 	\gamma_N(k,l) = \left|\u_{k,N}^* \left(\d_{1,N}\d_{2,N}^* + \d_{2,N}\d_{1,N}^*\right) \u_{\ell,N}\right|^2.
% 	%2 \Re\left(\eta_{k,N}^{(1,2)} \eta_{\ell,N}^{(1,2)}\right) + \frac{\eta_{k,N}^{(1,1)}\eta_{\ell,N}^{(2,2)}+ \eta_{\ell,N}^{(1,1)}\eta_{k,N}^{(2,2)}
% 	\notag
% \end{align}
%where $\eta^{(i,j)}_{k,N} = \d_{i,N}^* \u_{k,N}\u_{k,N}^* \d_{j,N}$. 
Let $\vartheta_{k,\ell}$ be defined under Assumption \ref{assumption:spiked} by
\begin{align}
	\vartheta_{k,\ell}=
	\frac
	{\sigma^4 c \left(\lambda_{k} \lambda_{\ell} + (\lambda_{k}+\lambda_{\ell})\sigma^2 + \sigma^4\right)\left(\lambda_{k} \lambda_{\ell} + \sigma^4 c\right)}
	{4\left(\lambda_{k}^2 - \sigma^4 c\right)\left(\lambda_{\ell}^2 - \sigma^4 c\right)\left(\lambda_{k}\lambda_{\ell} - \sigma^4 c\right)}
	\notag
\end{align}
for $1 \leq k,\ell \leq K$, and by
\begin{align}
	\vartheta_{k,\ell}= \frac{\sigma^2\left(\lambda_{k} + \sigma^2\right)}{4 \left(\lambda_{k}^2 - \sigma^4 c\right)}
	\notag
\end{align}
for $k \leq K, \ell \geq K+1$, with $\vartheta_{k,\ell}= \vartheta_{\ell,k}$, and set $\vartheta_{k,\ell} = 0$ for $k,\ell \geq K+1$.
Define finally
\begin{align}
	\gamma_N = \sum_{k,\ell=1}^M \vartheta_{k,\ell} \left|\u_{k,N}^* \left(\d_{1,N}\d_{2,N}^* + \d_{2,N}\d_{1,N}^*\right) \u_{\ell,N}\right|^2.
	\label{def:Gamma}
\end{align}
We then have the following result.
\begin{theorem}
	\label{theorem:clt_subspace}
	Under Assumption \ref{assumption:spiked}, if $\liminf_N \gamma_N > 0$, then
	\begin{align}
		\sqrt{N} \frac{\Re\left(\hat{\eta}_N - \eta_N\right)}{\sqrt{\gamma_N}} \xrightarrow[N \to \infty]{\Dcal} \Ncal_{\Rbb}\left(0, 1\right).
		\label{eq:conv_distrib_eta}
	\end{align}
\end{theorem}
The proof of Theorem \ref{theorem:clt_subspace}, which requires the use of technical tools from random matrix theory, is not included in the paper and is available in \cite{arxiv_clt}.

To conclude this section, we also provide a result on the asymptotic Gaussianity of the classical subspace estimator \eqref{eq:eta_hat_t}, which will prove to be useful to study the behaviour of MUSIC in the next section.
In the same way as \eqref{def:Gamma}, we define
\begin{align}
	\vartheta^{(t)}_{k,\ell} 
	= \frac{\sigma^4 c}{4} \frac{\chi_{k,\ell}^{(t)}}{\lambda_k \lambda_{\ell} (\lambda_k+\sigma^2 c)^2 (\lambda_{\ell} + \sigma^2 c)^2 (\lambda_k \lambda_{\ell} - \sigma^4 c)}
	\notag
\end{align}
for $1 \leq k,\ell \leq K$, where $\chi^{(t)}(k,\ell)$ is given at the top of the page (note that $\chi^{(t)}_{k,l} > 0$),
and by
\begin{align}
	\vartheta^{(t)}_{k,\ell}= \frac{\sigma^2\left(\lambda_{k} + \sigma^2\right)\left(\lambda_k^2 - \sigma^4 c\right)}{4 \lambda_k^2 \left(\lambda_{k} + \sigma^2 c\right)^2}.
	\notag
\end{align}
for $k \leq K, \ell \geq K+1$, with $\vartheta^{(t)}_{k,\ell}= \vartheta^{(t)}_{\ell,k}$, and set $\vartheta^{(t)}_{k,\ell} = 0$ for $k,\ell \geq K+1$.
Define finally
\begin{align}
	\gamma^{(t)}_N = \sum_{k,\ell=1}^M \vartheta^{(t)}_{k,\ell} \left|\u_{k,N}^* \left(\d_{1,N}\d_{2,N}^* + \d_{2,N}\d_{1,N}^*\right) \u_{\ell,N}\right|^2.
	\label{def:Gamma_t}
\end{align}
Then the following result holds.
\begin{theorem}
	\label{theorem:clt_subspace_trad}
	Under Assumption \ref{assumption:spiked}, if $\liminf_N \gamma_N^{(t)} > 0$, then
	\begin{align}
		\sqrt{N} \frac{\Re\left(\hat{\eta}^{(t)}_N - \eta^{(t)}_N\right)}{\sqrt{\gamma^{(t)}_N}} \xrightarrow[N \to \infty]{\Dcal} \Ncal_{\Rbb}\left(0, 1\right),
		\label{eq:conv_distrib_eta_t}
	\end{align}
	where 
	\begin{align}
		\eta^{(t)}_N = 
		\d_{1,N}^* 
		\left(
			\I - \sum_{k=1}^K \frac{\lambda_{k,N}^2 - \sigma^2 c_N}{\lambda_{k,N}\left(\lambda_{k,N}+\sigma^2 c_N\right)} \u_{k,N}\u_{k,N}^*
		\right)
		\d_{2,N}.
		\notag
	\end{align}
\end{theorem}
The proof of Theorem \ref{theorem:clt_subspace_trad} is given in \cite{arxiv_clt}.

	\subsection{Connections with others improved subspace estimators}
	\label{section:connection}	
	
Estimator \eqref{eq:noise_subspace_estimator} is valid under the hypothesis that the number of sources $K$ remains fixed when $N \rightarrow +\infty$. 
We recall that, under the hypothesis that the source signals are deterministic, or equivalently in the conditional case, 
\cite{vallet2012improved} proposed a consistent estimator of $\eta_N$, say $\hat{\eta}_{N,c}$, valid whatever $K$ is, and that it was proved in \cite{vallet2012improved} that 
\begin{align}
		\label{eq:cv-spike-conditional}
		\hat{\eta}_{N,c} - \hat{\eta}_N \rightarrow 0 \; \text{a.s.}  
\end{align}
It is even established in \cite[Remark 3.]{arxiv_clt} that 
\begin{equation}
	\label{eq:cv-spike-conditional-error-rate}
	\hat{\eta}_{N,c} - \hat{\eta}_N = o_{\Pbb}\left(\frac{1}{\sqrt{N}}\right)
\end{equation}
Therefore, if $K$ is fixed, the original subspace estimator derived in \cite{vallet2012improved} appears to be equivalent to the estimator \eqref{eq:noise_subspace_estimator}.

If the $K$ dimensional source signal $(\s_n)$ is assumed to be i.i.d. complex Gaussian, or equivalently in the unconditional case, \cite{mestre2008modified} proposed another consistent estimator, denoted $\hat{\eta}_{N,u}$, also valid whatever $K$ is in the unconditional case. 
When $K$ is fixed, and when $(\s_n)$ is deterministic, that is, in the conditional case, it is shown in the Appendix \ref{appendix:connections-mestre} that 
\begin{align}
	\hat{\eta}_{N,u} - \hat{\eta}_N = o_{\Pbb}\left(\frac{1}{\sqrt{N}}\right)
	\label{eq:cv-spike-unconditional}
\end{align}
Therefore, if $K$ is fixed, the subspace estimator of \cite{mestre2008modified}, in principle valid in the unconditional case, behaves as $\hat{\eta}_N$, or equivalently as the estimator $\hat{\eta}_{N,c}$ derived in \cite{vallet2012improved} in the 
conditional case. In conclusion, if $K$ is fixed, in the conditional case, the estimators $\hat{\eta}_{N,u}$, $\hat{\eta}_{N,c}$ and $\hat{\eta}_N$ are all equivalent. 
In section \ref{section:simu}, simulations are provided to illustrate that $\hat{\eta}_N$ and $\hat{\eta}_{N,u}(\theta)$ present the same performance, in the context of DoA estimation.

\section{Analysis of the consistency of G-MUSIC and MUSIC}
\label{section:consistency}

From now on, we use the results of section \ref{section:review_rmt} for $\Sigmabs_N= N^{-1/2}\Y_N$, $\B_N=N^{-1/2}\A(\thetabs)\S_N$, $\W_N=N^{-1/2}\V_N$, $\d_{1,N} = \d_{2,N} = \a(\theta)$ and assume that Assumption \ref{assumption:spiked} holds. 
Based on the subspace estimator \eqref{eq:noise_subspace_estimator},  \cite{hachem2012subspace} proposed the improved pseudo-spectrum estimator
\begin{align}
	\hat{\eta}_N(\theta) = 1 -  \sum_{k=1}^K \frac{1}{h\left(\hat{\lambda}_{k,N} \right)} \left|\a(\theta)^*\hat{\u}_{k,N}\right|^2,
	\label{eq:ps_estimator}
\end{align}

\begin{remark}
	The pseudo-spectrum estimator \eqref{eq:ps_estimator} can be viewed as a weighted version of the traditional pseudo-spectrum estimator
	\begin{align}
		\hat{\eta}_N^{(t)}(\theta) = 1 - \sum_{k=1}^K \left|\a(\theta)^*\hat{\u}_{k,N}\right|^2.
		\notag
	\end{align}
	Therefore, there is no additional computational cost by using this improved pseudo-spectrum estimator (which gives the G-MUSIC method described below), 
	since it also relies on an eigenvalues/eigenvectors decomposition of the SCM $\frac{1}{N} \Y_N\Y_N^*$.
	Moreover, in the traditional asymptotic regime where $\frac{M}{N} \to 0$, by setting $c=0$, we remark that $h(z)=1$ and 
	thus the improved pseudo-spectrum estimator reduces to the traditional one.
\end{remark}
From \eqref{eq:consistency_subspace_estimator}, we have directly that $\hat{\eta}_N(\theta)-\eta_N(\theta) \to 0$ a.s. as $N \to \infty$, for all $\theta$.
In Hachem et al. \cite{hachem2012large}, this convergence was also proved to be uniform, that is 
\begin{align}
	\sup_{\theta \in [-\pi,\pi]} \left|\hat{\eta}_N(\theta)-\eta_N(\theta)\right| \xrightarrow[N \to \infty]{a.s.} 0,
	\label{eq:uniform_consistency}
\end{align}
The resulting DoA estimation method, termed as G-MUSIC, consists in estimating $\theta_1,\ldots,\theta_K$ as the $K$ most significant minima of $\theta \mapsto \hat{\eta}_N(\theta)$.
%
%\begin{remark}
%	The G-MUSIC method is based on the subspace estimator \eqref{eq:noise_subspace_estimator}, which corresponds to the ``spiked model" approximation 
%	($K$ fixed with respect to $N$) of the subspace estimator originally derived in \cite[eq. (27)]{vallet2012improved} 
%	(see also section \ref{section:connection}	 of this reference)
%	in the context of $K$ possibly going to infinity with $N$.
%	It turns out that these two subspace estimators present the same asymptotic performance under Assumption \ref{assumption:spiked}, in terms of consistency
%	and asymptotic Gaussianity (Theorem \ref{theorem:clt_subspace}). Therefore, we choose in this paper to present the spiked model version
%	\eqref{eq:noise_subspace_estimator}, to simplify the presentation.
%\end{remark}

Concerning the traditional pseudo-spectrum estimator 
$\hat{\eta}_N^{(t)}(\theta)$, Theorem \ref{theorem:spiked_eigv} directly
implies that for all $\theta$,
\begin{align}
	\hat{\eta}_N^{(t)}(\theta) - \eta_N^{(t)}(\theta) 
	\xrightarrow[N\to\infty]{a.s.} 0,
	\notag
\end{align}
where 
\begin{align}
	\eta_N^{(t)}(\theta) = 
	1 - \sum_{k=1}^K 
	\frac
	{
		\lambda_k^2 - \sigma^4 c
	}
	{\lambda_k\left(\lambda_k+\sigma^2 c\right)}
	\left|\a(\theta)^*\u_{k,N}\right|^2.
	\label{eq:music_asymp_eq}
\end{align} 
 
% In the next subsections, we review the consistency of the G-MUSIC DoA estimates in the context of widely spaced DoA, that is, $\theta_1,\ldots,\theta_K$ fixed with respect to $N$,
% and establish the consistency of the MUSIC DoA estimates in this context. We also use the CLT \ref{theorem:clt_subspace} to prove the asymptotic Gaussianity of the G-MUSIC DoA estimates,
% and show that its asymptotic MSE coincides with the one of MUSIC.

	\subsection{$N$-consistency for widely spaced DoA}
	\label{section:consistency_wide}
	
In this section, we consider a widely spaced DoA scenario. In practice, such a situation occurs e.g. when the DoA have an angular separation
much larger than a beamwidth $\frac{2 \pi}{M}$. Mathematically speaking, we will therefore consider that the DoA $\theta_1,\ldots,\theta_K$ are fixed with 
respect to $N$.
In that case, $\A^* \A \to \I$ and the separation condition (Assumption \ref{assumption:spiked}) holds if and only if the eigenvalues of $\frac{\S_N\S_N^*}{N}$
converge to $\lambda_1 > \ldots > \lambda_K > \sigma^2 \sqrt{c}$.
To summarize, we make the following assumption.
\begin{assum}
	\label{assumption:fixed_DoA}
	$K$, $\theta_1,\ldots,\theta_K$ are independent of $N$, and the eigenvalues of $N^{-1}\S_N\S_N^*$ converge to 
	\begin{align}
		\lambda_1 > \ldots > \lambda_K > \sigma^2 \sqrt{c}.
		\notag
	\end{align}
\end{assum}
Note that Assumption \ref{assumption:fixed_DoA} allows in particular spatial correlation between sources, since $\frac{1}{N} \S_N\S_N^*$ may converge to a positive definite matrix, which is not necessarily constrained to be diagonal. 

To study the consistency of G-MUSIC and MUSIC, we need to define ``properly" the corresponding estimators, to avoid identifiability issues. 
As it is usually done in the theory of M-estimation, we consider $\Ical_1,\ldots,\Ical_K \subset [-\pi,\pi]$ $K$ compact disjoint intervals 
such that $\theta_k \in \mathrm{Int}\left(\Ical_k\right)$ ($\mathrm{Int}$ denotes the interior of a set), and formally define the G-MUSIC and MUSIC DoA estimators as
\footnote
{
	Note that the G-MUSIC cost function 
	%needs a modulus since for finite $N$, 
	%it 
	can be negative due to the presence of the weighting factor 
	$h(\hat{\lambda}_{k,N})^{-1}$ in \eqref{eq:ps_estimator}. 
	%However, for the asymptotic analysis above, this modulus is not necessary. 
}
\begin{align}
	\hat{\theta}_{k,N} = \argmin_{\theta \in \Ical_k} \hat{\eta}_N(\theta)
	\quad\text{and}\quad
	\hat{\theta}^{(t)}_{k,N} = \argmin_{\theta \in \Ical_k} \hat{\eta}_N^{(t)}(\theta).
	\label{def:Doa_estimates_K_fixed}
\end{align}

We have the following result, whose proof  is deferred to Appendix \ref{appendix:proof_consistency_low}.
\begin{theorem}
	\label{theorem:consistency_MUSIC_GMUSIC_low}
	Under Assumption \ref{assumption:fixed_DoA}, for $k=1,\ldots,K$,
	\begin{align}
		\hat{\theta}_{k,N} = \theta_k + o\left(\frac{1}{N}\right)
		\quad\text{and}\quad
		\hat{\theta}^{(t)}_{k,N} = \theta_k + o\left(\frac{1}{N}\right),
		\notag
	\end{align}
	with probability one.
\end{theorem}

The results of Theorem \ref{theorem:consistency_MUSIC_GMUSIC_low} show that both the G-MUSIC and MUSIC methods have the same first order behaviour, i.e. are
$N$-consistent, when the angles $\theta_1,\ldots,\theta_K$ are fixed with respect to $N$. 
In section \ref{section:clt}, it will be further shown that the MUSIC method also has the same asymptotic MSE as the G-MUSIC method as $N \to \infty$.

	\subsection{$N$-consistency for closely spaced DoA}

In this section, we study the consistency of G-MUSIC and MUSIC in a closely spaced DoA scenario, where we let the DoA $\theta_{1,N},\ldots,\theta_{K,N}$
depends on $N$ and converge to the same value at rate $\Ocal\left(\frac{1}{M}\right)$. 
To simplify the presentation, we only consider $K=2$ sources with DoA $\theta_{1,N}$ and $\theta_{2,N}=\theta_{1,N} + \frac{\alpha}{N}$, where $\alpha > 0$, and 
assume asymptotic uncorrelated sources with equal powers, that is $N^{-1} \S_N\S_N^* \to \I$.
In this case, it is easily seen that the two non null signal eigenvalues of $\frac{\A\S_N\S_N^*\A^*}{N}$ converge to
\begin{align}
	\lambda_{1}(\alpha) = 1+\left|\sinc\left(\frac{\alpha c}{2}\right)\right|
	\quad\text{and}\quad
	\lambda_{2}(\alpha) = 1-\left|\sinc\left(\frac{\alpha c}{2}\right)\right|.
	\notag
\end{align}
where $\sinc(x) = \sin(x)/x$ if $x \neq 0$ and $\sinc(0)=1$. 
Therefore, the subspace separation condition (Assumption \ref{assumption:spiked}) holds if and only if $\lambda_2(\alpha) > \sigma^2 \sqrt{c}$.
To summarize, we consider the following assumption.
\begin{assum}
	\label{assumption:close_DoA}
	We assume that $K=2$, 
	\begin{align}
		\frac{\S_N\S_N^*}{N} \xrightarrow[N\to\infty]{} \I,
		\notag
	\end{align}
	and that the DoA $\theta_{1,N},\theta_{2,N}$ depend on $N$ in such a way that
	\begin{align}
		\theta_{2,N} = \theta_{1,N} + \frac{\alpha}{N},
		\notag
	\end{align}
	where $\alpha > 0$ satisfies
	\begin{align}
		\left|\sinc\left(\frac{\alpha c}{2}\right)\right| < 1 - \sigma^2 \sqrt{c}.
		\notag
	\end{align}
\end{assum}
Since the DoA are not fixed with respect to $N$, we define, in the same way as \eqref{def:Doa_estimates_K_fixed}, the G-MUSIC and MUSIC DoA estimates as
\begin{align}
	\hat{\theta}_{k,N} = \argmin_{\theta \in \Ical_{k,N}} \hat{\eta}_N(\theta)
	\quad\text{and}\quad
	\hat{\theta}_{k,N}^{(t)} = \argmin_{\theta \in \Ical_{k,N}} \hat{\eta}_N^{(t)}(\theta)
\end{align}
where $\Ical_{k,N}$ is defined as the compact interval 
\begin{align}
	\Ical_{k,N} = \left[\theta_{k,N}-\frac{\alpha-\epsilon}{2N},\theta_{k,N}+\frac{\alpha -\epsilon}{2N}\right],
	\notag
\end{align}
with $0 < \epsilon < \alpha$.
The $N$-consistency results for G-MUSIC and MUSIC in the closely spaced DoA scenario can be summarized as follows.
\begin{theorem}
	\label{theorem:consistency_close}
	Under Assumption \ref{assumption:close_DoA}, for $k \in \{1,2\}$,
	\begin{align}
		\hat{\theta}_{k,N} = \theta_{k,N} + o\left(\frac{1}{N}\right),
		\label{eq:DoA_consistency_high_resolution}
	\end{align}
	with probability one.
	Moreover, if $0$ and $\alpha$ are not local maxima of the function $\beta \mapsto \kappa^{(t)}(\beta)$ defined by
	\begin{align}
		&\kappa^{(t)}(\beta) = 
		\notag\\		
		&\frac{\left(\lambda_1(\alpha)^2 - \sigma^4 c\right)\left(\sinc(\beta c /2) + \sinc((\beta-\alpha)c/2)\right)^2}
		{2 \lambda_1(\alpha)^2 \left(\lambda_1(\alpha) + \sigma^2 c\right)}
		\notag\\
		&+\frac{\left(\lambda_2(\alpha)^2 - \sigma^4 c\right)\left(\sinc(\beta c /2) - \sinc((\beta-\alpha)c/2)\right)^2}
		{2 \lambda_2(\alpha)^2 \left(\lambda_2(\alpha) + \sigma^2 c\right)},
		\label{eq:kappat}
	\end{align}
	then $N\left(\hat{\theta}_{k,N}^{(t)} - \theta_{k,N}\right)$ does not converge to $0$.
\end{theorem}
The proof of Theorem \ref{theorem:consistency_close} is deferred to Appendix \ref{appendix:proof_consistency_close}.

Theorem \ref{theorem:consistency_close} shows that the G-MUSIC method remains $N$-consistent when two sources have DoA with a spacing of the order $\Ocal\left(M^{-1}\right)$ while MUSIC may not be able to consistently separate the two DoA if the spacing parameter $\alpha$ is not a local maximum of the function defined in \eqref{eq:kappat} (numerical examples are given in Figure \ref{fig:kappat}). This confirms the superiority of G-MUSIC over MUSIC in closely spaced DoA situations and low sample size situations.

\begin{figure}
	\centering
	\subfigure[$\alpha=0.25\pi/c$]{\includegraphics[scale=0.4]{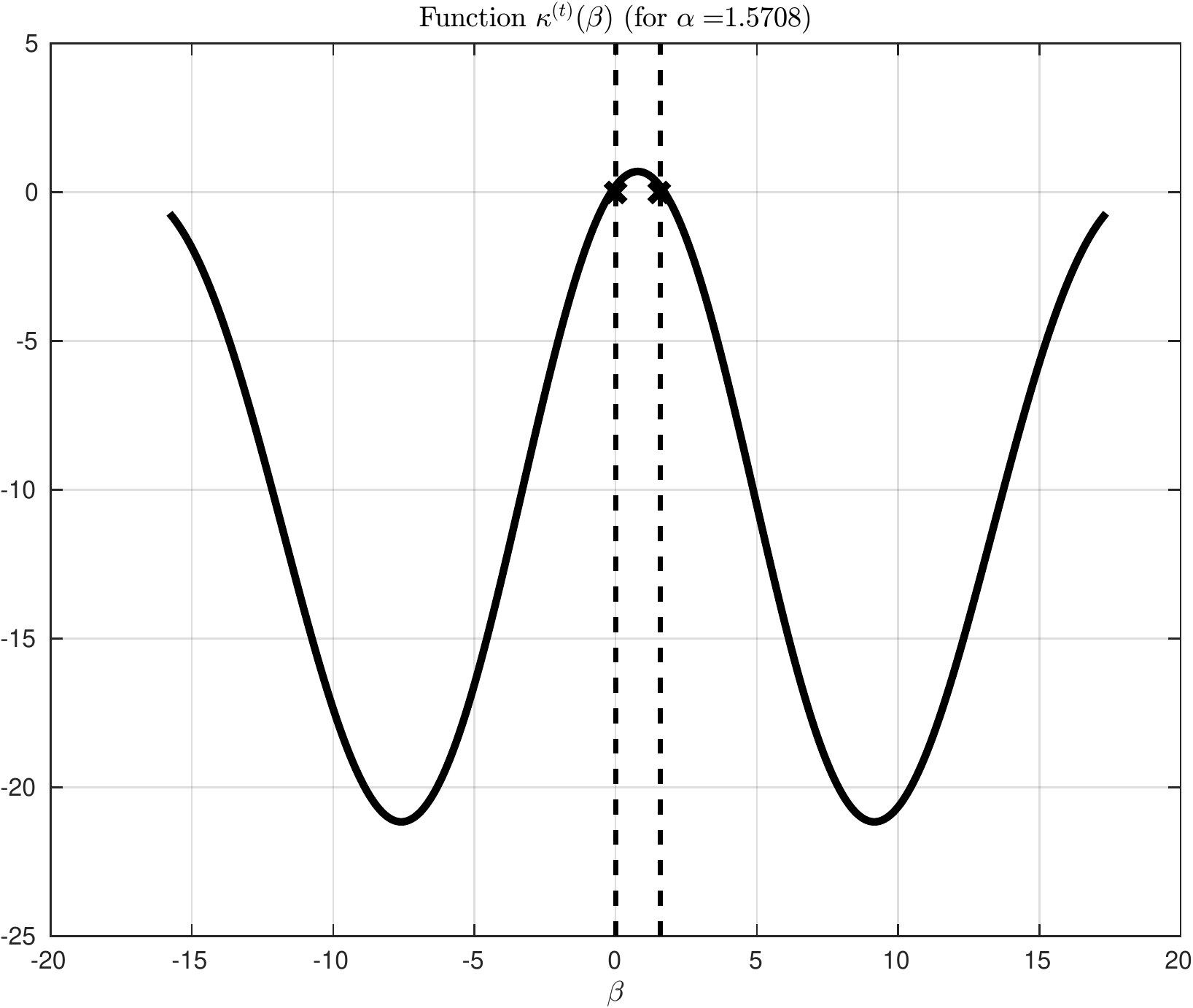}}
	\subfigure[$\alpha=2\pi/c$]{\includegraphics[scale=0.4]{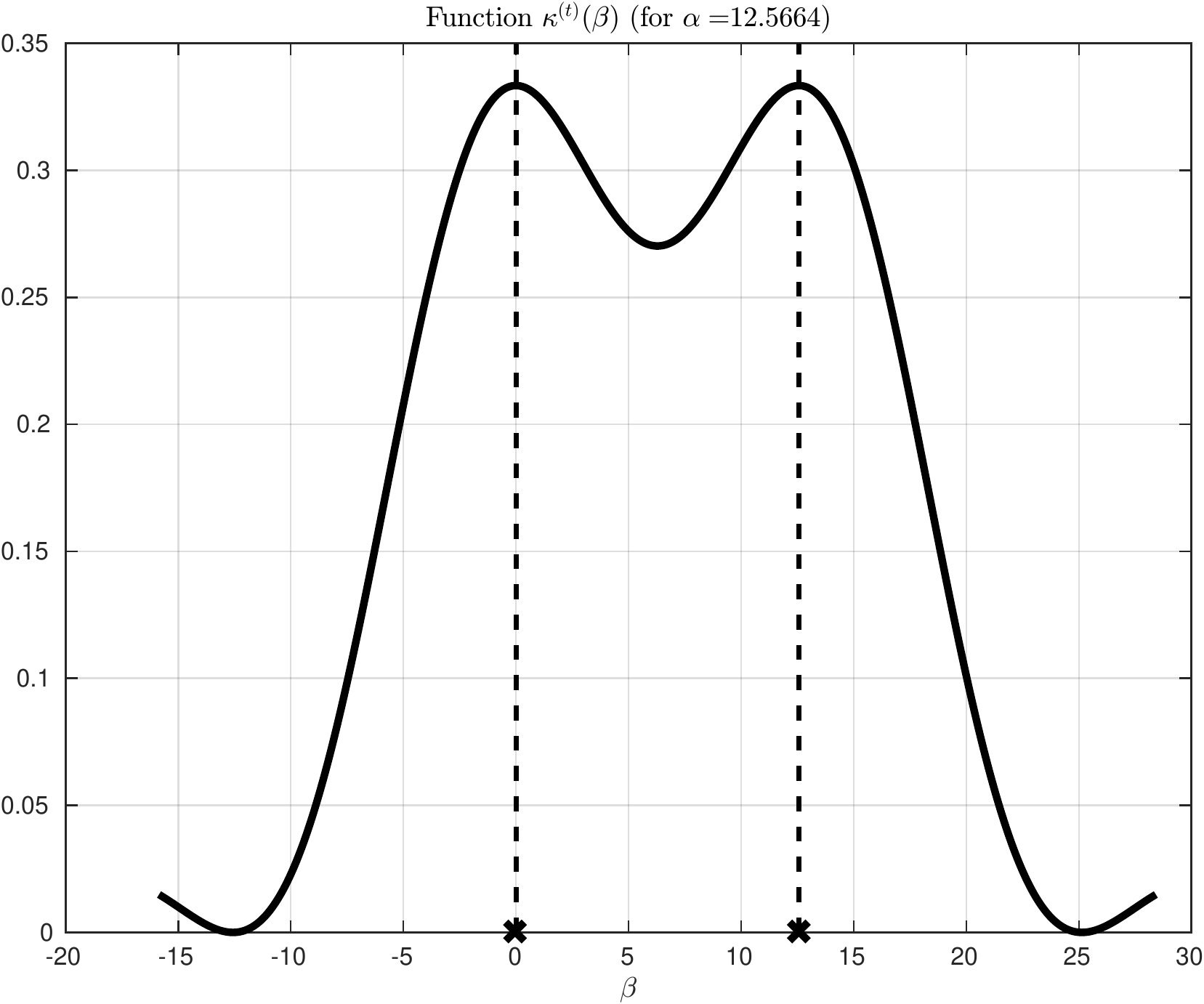}}
	\caption{Function $\beta \mapsto \kappa^{(t)}(\beta)$ for $\sigma=1$, $c=0.5$ and different values of $\alpha$ (the dashed lines represent the location
	of $0$ and $\alpha$)}
	\label{fig:kappat}
\end{figure}

\subsection{Remarks on the spatial periodogram}
\label{section:spatial_periodogram}

Regarding the previous results on the consistency of the MUSIC estimator for widely spaced and closely spaced scenarios, it is natural to ask how
traditional ``low resolution" techniques for DoA estimation behave.

Considering the classical spatial periodogram cost function, that is
\begin{align}
	\hat{\eta}_N^{(p)} (\theta) = \a(\theta)^* \frac{\Y_N\Y_N^*}{N} \a(\theta),
	\notag
\end{align}
we can prove, such as in \cite[Sec. 3.3]{hachem2012large}, that
\begin{align}
	\sup_{\theta \in [-\pi,\pi]} \left|\hat{\eta}_N^{(p)} (\theta) - \eta_N^{(p)}(\theta)\right| 	\xrightarrow[N\to\infty]{a.s.} 0,
	\notag
\end{align}
where $\eta_N^{(p)}(\theta) = \a(\theta)^* \left(\A \frac{\S_N\S_N^*}{N}\A^* + \sigma^2 \I\right) \a(\theta)$.
Moreover, following the steps of the proof of Theorem \ref{theorem:consistency_MUSIC_GMUSIC_low} for the MUSIC estimates, we end up as well with
\begin{align}
	\hat{\theta}^{(p)}_{k,N} = \theta_{k} + o\left(\frac{1}{N}\right),
	\notag
\end{align}
with probability one, where $\hat{\theta}^{(p)}_{k,N} = \argmax_{\theta \in \Ical_k} \hat{\eta}_N^{(p)} (\theta)$. 
Therefore, the spatial periodogram also provides consistent estimate in the widely-spaced DoA scenario, without any requirements on the sources power (i.e. without need of the separation condition $\lambda_{k} > \sigma^{2} \sqrt{c}$, $k=1, \ldots, K$).
This confirms the well-known fact that the use of subspace methods, especially MUSIC, is not necessarily a relevant choice for estimating the DoA of widely spaced sources. Nevertheless, in certain scenarios involving correlated source signals and widely spaced DoA and for which the spatial periodogram may exhibit a non negligible bias at high SNR, the use of subspace methods may still be interesting (see numerical illustrations in Section \ref{section:simu}).

However, in the scenario of closely spaced DoA, one can also prove, following the steps of Theorem \ref{theorem:consistency_close}, that the spatial periodogram
suffers the same drawback as MUSIC, and is not capable of consistently separating two DoA with an angular spacing of the order $\Ocal\left(\frac{1}{M}\right)$.
Simulations are provided in the next section to illustrate these facts.

	\section{Asymptotic Gaussianity of G-MUSIC and MUSIC}
	\label{section:clt}
		
We now apply the results of Theorems \ref{theorem:clt_subspace} and \ref{theorem:clt_subspace_trad} to obtain a Central Limit Theorem for the G-MUSIC and MUSIC DoA estimates. 
The results for G-MUSIC will be valid for both the widely spaced and closely spaced DoA scenarios introduced in the previous section, 
while the CLT for MUSIC will be only valid for the widely spaced DoA scenario, since it is not $N$-consistent for the other situation.

		\subsection{CLT for G-MUSIC and MUSIC}
		\label{section:clt_MUSIC_GMUSIC}
		
The following Theorem, whose proof is given in Appendix \ref{appendix:proof_clt_GMUSIC}, provides the asymptotic Gaussianity of the G-MUSIC DoA estimates, 
under Assumption \ref{assumption:fixed_DoA} or Assumption \ref{assumption:close_DoA}. 
We denote by $\a'(\theta)$ and $\a''(\theta)$ respectively the first and second order derivatives w.r.t. $\theta$ of the function $\theta \rightarrow \a(\theta)$. 
\begin{theorem}
	\label{theorem:clt_GMUSIC}
	%Assume that Assumption \ref{assumption:fixed_DoA} holds, or that 	
	%Assumption \ref{assumption:close_DoA} with 
	%$\liminf_N \left\|\Pibs_N \d_{1,N}\right\|^2 > 0$ holds. 
	Under Assumption \ref{assumption:fixed_DoA} or Assumption \ref{assumption:close_DoA},
	\begin{align}
		N^{3/2} \sqrt{\frac{ \left(\d_{1,N}^* \Pibs_N \d_{1,N}\right)^2 }{\gamma_N}}
		\left(\hat{\theta}_{k,N} - \theta_{k,N}\right) 
		\xrightarrow[N\to\infty]{\Dcal} \Ncal_{\Rbb}(0,1).
		\label{eq:CLT_G_MUSIC}
	\end{align}
	for $k=1,\ldots,K$, where $\gamma_N$ is defined by \eqref{def:Gamma}, with 
	$\d_{1,N}=N^{-1}\a'(\theta_{k,N})$, $\d_{2,N}=\a(\theta_{k,N})$.
\end{theorem}
%Of course, Theorem \ref{theorem:clt_GMUSIC} holds under Assumption
% \ref{assumption:fixed_DoA} or Assumption \ref{assumption:close_DoA}, that is under the
% widely and closely spaced DoA scenarios introduced in Section \ref{section:consistency}. 
In particular, by considering the settings of Assumption \ref{assumption:fixed_DoA} and adding the following spatial uncorrelation condition
\begin{align}
	\frac{\S_N\S_N^*}{N} \xrightarrow[N\to\infty]{} 
	\diag\left(\lambda_1,\ldots,\lambda_K\right),
	\label{eq:spatial_uncorrelation}
\end{align}
we obtain, using the usual asymptotic orthogonality between $\a(\theta_k)$ and $\u_{k',N}$ for $k \neq k'$ (see e.g. \cite[Lem. 8]{hachem2012subspace}),
\begin{align}
	\d_{1,N}^* \Pibs_N \d_{1,N}\xrightarrow[N\to\infty]{} \frac{c^2}{12} 
	\text{ and } 
	\gamma_N \xrightarrow[N\to\infty]{} \frac{c^2}{24} \frac{\sigma^2 (\lambda_k+\sigma^2)}{\lambda_k^2 - \sigma^4 c}.
	\notag
\end{align}
Thus, we retrieve the results of \cite{hachem2012subspace} under this particular assumption:
\begin{align}
		N^{3/2} \left(\hat{\theta}_{k,N} - \theta_{k}\right) 
		\xrightarrow[N\to\infty]{\Dcal} \Ncal_{\Rbb}\left(0,\frac{6}{c^2} \frac{\sigma^2 (\lambda_k+\sigma^2)}{\lambda_k^2 - \sigma^4 c}\right),
		\label{eq:clt_GMUSIC_wide}
\end{align}
Therefore, Theorem \ref{theorem:clt_GMUSIC} extends the results of \cite{hachem2012subspace} to more general scenarios of correlated sources and not necessarily widely distributed sources.

Concerning the MUSIC method, we obtain the following result in the widely spaced DoA scenario.
\begin{theorem}
	\label{theorem:clt_MUSIC}
	Let $\eta_N^{(t) (2)} (\theta)$ be the second order derivative of 
	$\theta \mapsto \eta_N^{(t)}(\theta)$ defined in \eqref{eq:music_asymp_eq}.
	Under Assumption \ref{assumption:fixed_DoA}, and if 
	\begin{align}
		\liminf_{N \to \infty}  |\eta_N^{(t)(2)}(\theta_k)| > 0, 
		\notag
	\end{align}
	it holds that
	\begin{align}
		N^{3/2} \sqrt{\frac{ \eta_N^{(t)(2)}(\theta_k)^2}{4 \gamma_N^{(t)}}}\left(\hat{\theta}^{(t)}_{k,N} - \theta_{k} \right) 
		\xrightarrow[N\to\infty]{\Dcal} \Ncal_{\Rbb}\left(0,1\right),
		\label{eq:clt_MUSIC_wide}
	\end{align}
	for $k=1,\ldots,K$, where $\gamma_N^{(t)}$ is defined by \eqref{def:Gamma_t} by setting $\d_{1,N}=N^{-1}\a'(\theta_{k})$ and $\d_{2,N}=\a(\theta_{k})$, 	
% 	\begin{align}
% 		\begin{split}
% 		\eta_N^{(t)(2)}(\theta) &= 
% 		- 2 \sum_{k=1}^K h\left(\phi(\lambda_k)\right) 
% 		\Biggl(
% 		\left|\frac{\a'(\theta)}{N}^*\u_{k,N}\right|^2 
% 		\notag\\
% 		&\qquad\qquad \qquad
% 		+ \Re\left(\frac{\a''(\theta_k)}{N^2} \u_{k,N}\u_{k,N}^* \a(\theta_k)\right)
% 		\Biggr)
% 		\notag
% 		\end{split}
% 	\end{align}
\end{theorem}
The proof of Theorem \ref{theorem:clt_MUSIC}, which is based on the CLT of Theorem \ref{theorem:clt_subspace_trad}, is similar to the one of Theorem \ref{theorem:clt_GMUSIC} and is therefore omitted.

Theorem \ref{theorem:clt_MUSIC}, having been derived under Assumption \ref{assumption:fixed_DoA}, allows in particular correlation between source signals. 
%as long as the DoA
%$\theta_1,\ldots,\theta_K$ does not depend on $N$.
Moreover, by assuming asymptotic uncorrelation between sources, i.e. that \eqref{eq:spatial_uncorrelation} holds, we obtain
\begin{align}
	\frac{1}{N^2} \eta_N^{(t)(2)}(\theta_k) \xrightarrow[N\to\infty]{} \frac{c^2 (\lambda_k^2 - \sigma^4 c)}{6 \lambda_k (\lambda_k + \sigma^2 c)},
\end{align}
and
\begin{align}
	\gamma^{(t)}_N \xrightarrow[N\to\infty]{} \frac{c^2 (\lambda_k + \sigma^2)(\lambda_k^2 - \sigma^4 c)}{24 \lambda_k^2 (\lambda_k + \sigma^2 c)^2},
	\notag
\end{align}
which implies
\begin{align}
	N^{3/2} \left(\hat{\theta}_{k,N}^{(t)} - \theta_{k}\right) 
	\xrightarrow[N\to\infty]{\Dcal} \Ncal_{\Rbb}\left(0,\frac{6}{c^2} \frac{\sigma^2 (\lambda_k+\sigma^2)}{\lambda_k^2 - \sigma^4 c}\right).
	\label{eq:clt_MUSIC_wide2}
\end{align}
The striking fact about Theorem \ref{theorem:clt_MUSIC} is that, in the widely spaced scenario, the variance of the MUSIC estimates obtained in \eqref{eq:clt_MUSIC_wide2} coincides with the variance of the G-MUSIC estimates \eqref{eq:clt_GMUSIC_wide} previously derived in \cite{hachem2012subspace}. 
This shows that MUSIC and G-MUSIC present exactly the same asymptotic performance for widely spaced DoA and uncorrelated sources, which reinforces the conclusions given in Section \ref{section:consistency_wide}.

		\subsection{Numerical examples}
		\label{section:simu}		
		
In this section, we provide numerical simulations illustrating the results given in the previous sections.

To illustrate the similarity between the theoretical MSE (formula of Theorem \ref{theorem:clt_GMUSIC}) and its approximation for uncorrelated source signal
and widespace DoA (specific formula of \eqref{eq:clt_GMUSIC_wide}), we plot these two formulas in Figure \ref{figure:mse_wide_uncorr}
and Figure \ref{figure:mse_wide_corr}, together with the empirical MSE of the G-MUSIC estimate $\hat{\theta}_{1,N}$ and the Cramer-Rao bound (CRB). 
The parameters are  $K = 2$, $M=40$, $N=80$, $\mathrm{SNR}=-10 \log (\sigma^2)$.
In Figure \ref{figure:mse_wide_uncorr}, we consider the context of widespace DoA with uncorrelated source signal, by choosing a signal matrix $\S_N$ with standard i.i.d $\Ncal_{\Cbb}(0,1)$ entries, and setting $\theta_1=0$, $\theta_2=5 \times \frac{2 \pi}{M}$. The separation condition $\lambda_K > \sigma^2 \sqrt{c}$ occurs around SNR = 0 dB.
In this situation, we notice that the two MSE formulas match, as discussed in Section \ref{section:clt_MUSIC_GMUSIC}.
\begin{figure}[h!]
	\centering
	\subfigure[Uncorrelated source signals]{\label{figure:mse_wide_uncorr}\includegraphics[scale=0.4]{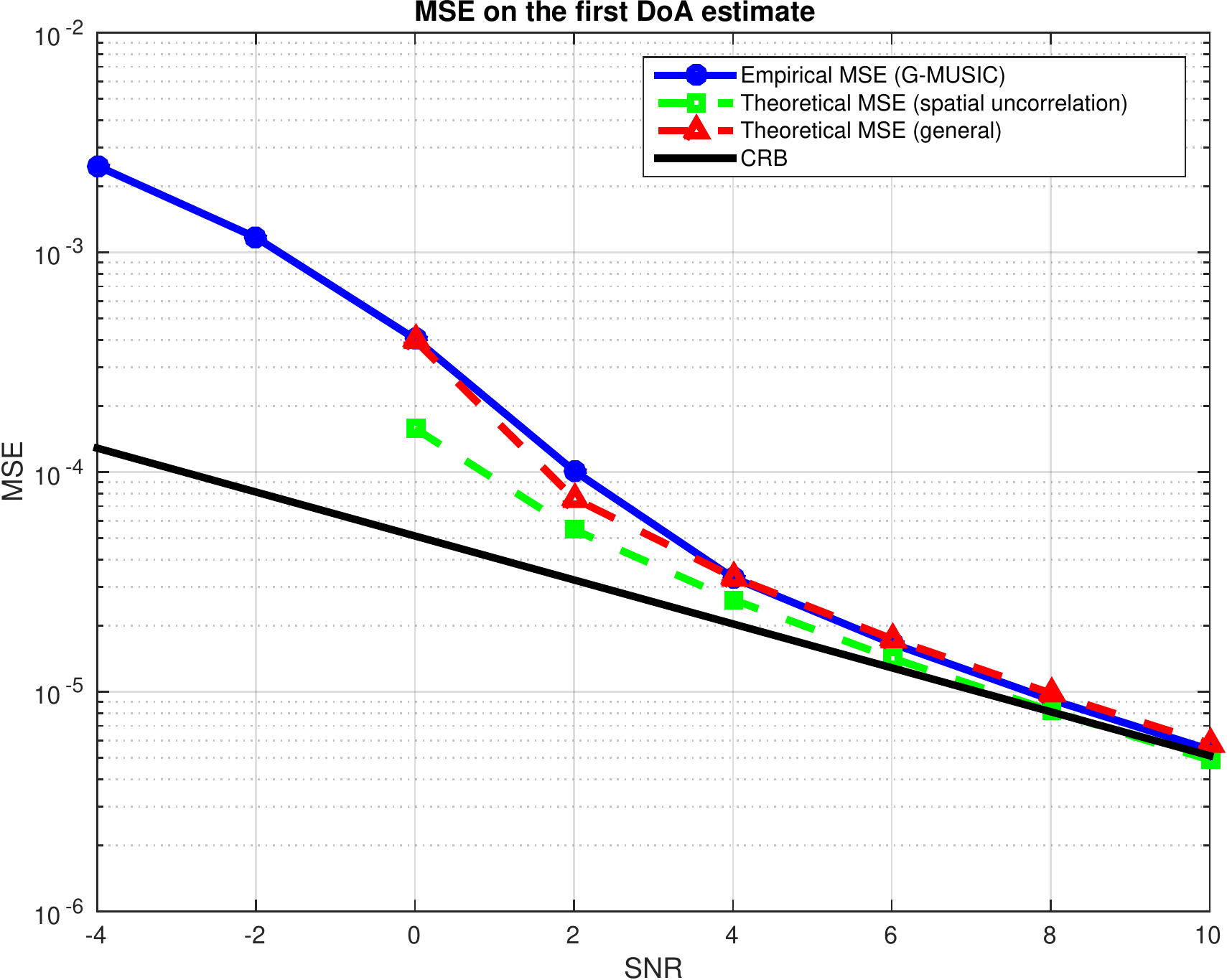}}
	\subfigure[Correlated source signals]{\label{figure:mse_wide_corr}\includegraphics[scale=0.4]{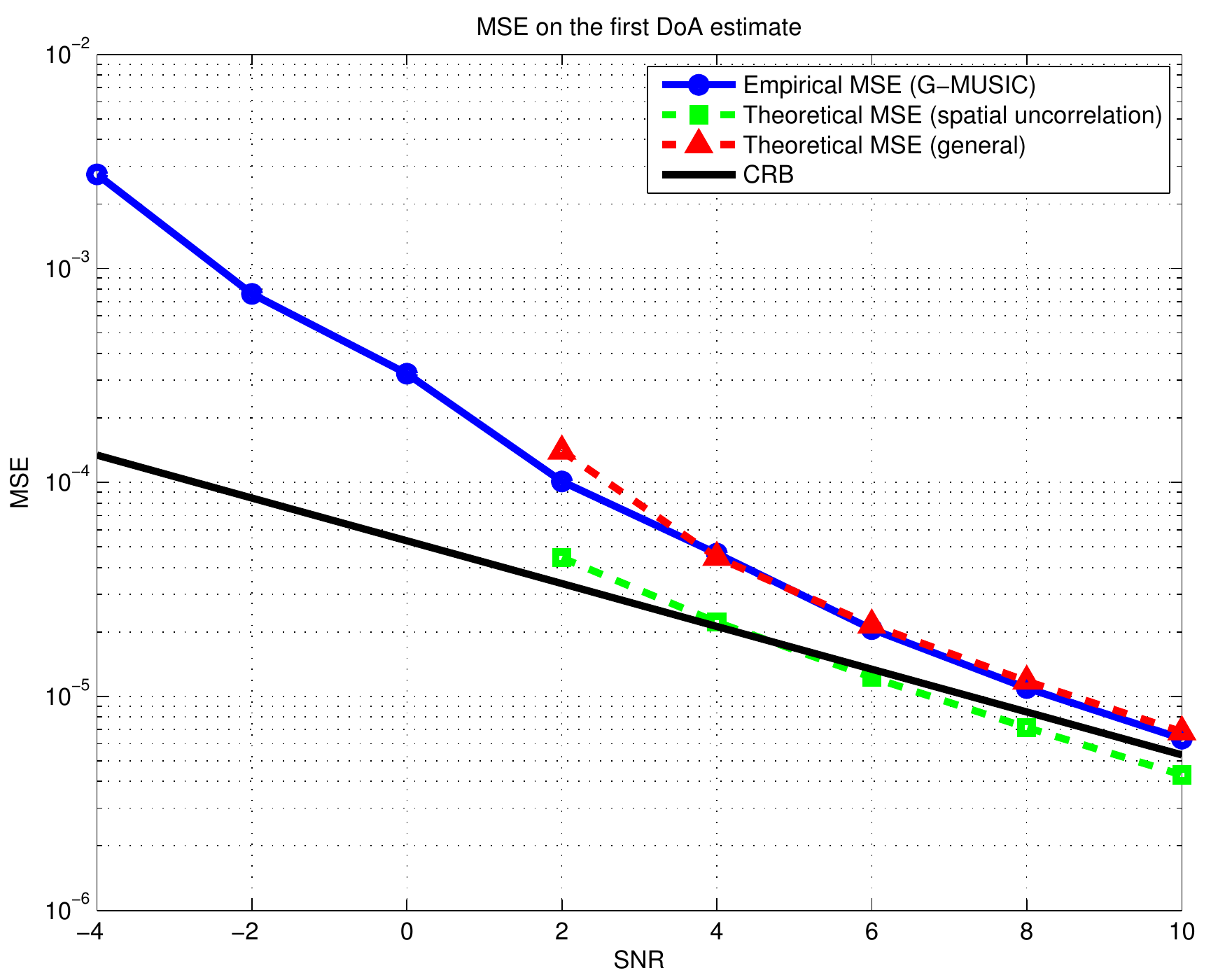}}
	\caption{Empirical MSE of $\hat{\theta}_{1,N}$ for widely spaced DoA versus SNR (dB)}
	\label{figure:mse_wide}
\end{figure}
In Figure \ref{figure:mse_wide_corr}, we consider the context of widespace DoA with significant correlation between source signals, by choosing a matrix $\S_N=\R^{1/2}\X_N$ with $\R=[1,0.4;0.4,1]$ and $\X_N$ having standard i.i.d $\Ncal_{\Cbb}(0,1)$ entries. 
The separation condition occurs around SNR = 2 dB.
We notice that the MSE formula of Theorem \ref{theorem:clt_GMUSIC} is relatively accurate while a discrepancy may occur 
for the formula \eqref{eq:clt_GMUSIC_wide}, since the spatial uncorrelation is not fulfilled in that case.

In Figure \ref{figure:mse_MUSIC_GMUSIC_wide_uncorr}, we consider the context of widespace DoA and uncorrelated source signals, and compare the performance of G-MUSIC, MUSIC and DoA estimation with spatial periodogram, in terms of MSE on the first DoA estimate.
The empirical MSE of $\hat{\theta}_{1,N}$ together with its theoretical MSE given in Theorem \ref{theorem:clt_GMUSIC}, as well as the empirical MSE of $\hat{\theta}_{1,N}^{(t)}$ and $\hat{\theta}^{(p)}_{1,N}$ are plotted. The parameters are $M=40$, $N=80$, and $\theta_1=0$, $\theta_2=5 \times \frac{2 \pi}{M}$. The signal matrix $\S_N$ has standard i.i.d $\Ncal_{\Cbb}(0,1)$ entries, and the separation condition occurs around SNR = 0 dB.
\begin{figure}[h!]
	\centering
	\includegraphics[scale=0.4]{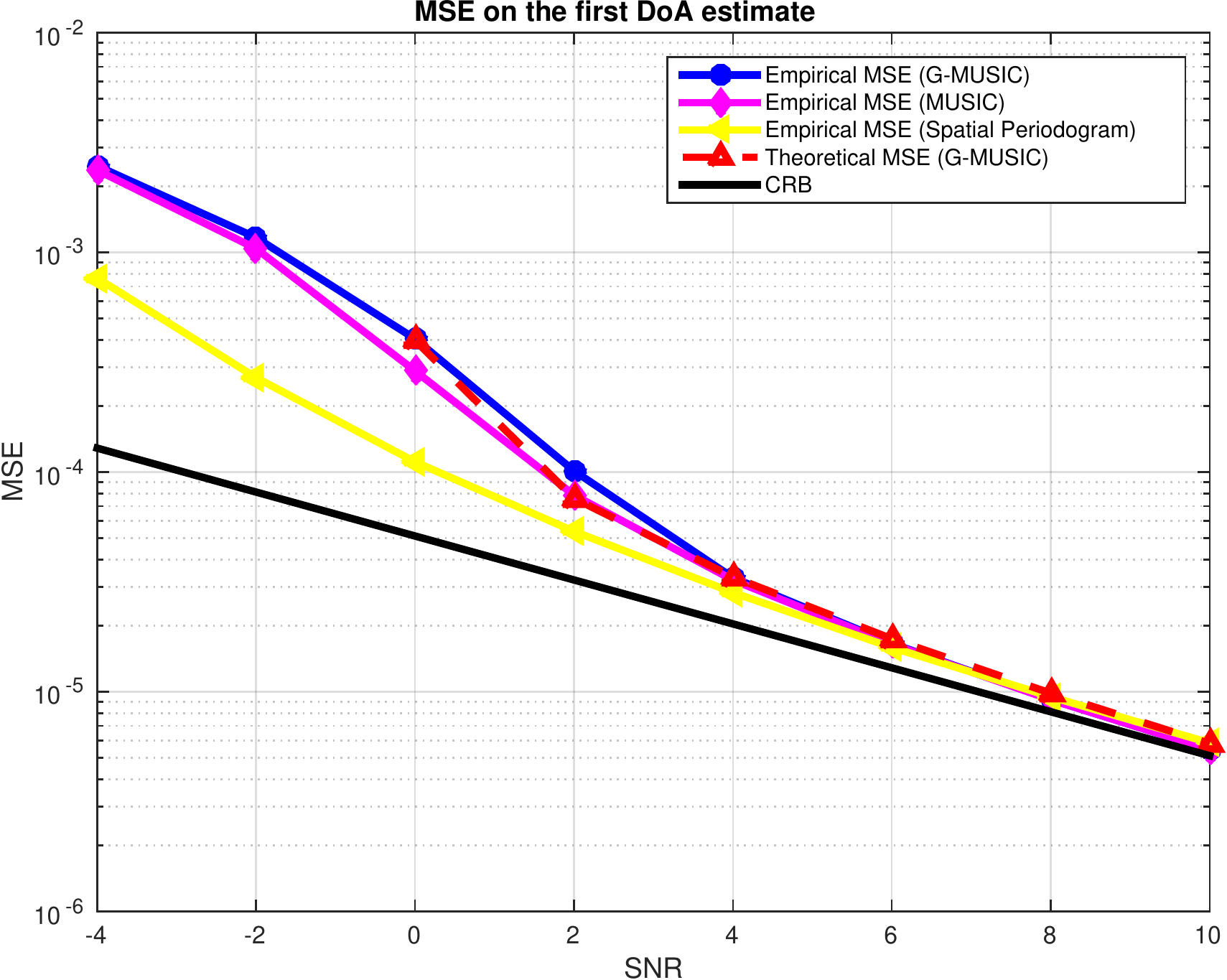}
	\caption
	{
		Empirical MSE of $\hat{\theta}_{1,N}$, $\hat{\theta}^{(t)}_{1,N}$ and $\hat{\theta}^{(p)}_{1,N}$ for widely spaced DoA and uncorrelated source signals,
		versus SNR (dB)
	}
	\label{figure:mse_MUSIC_GMUSIC_wide_uncorr}
\end{figure}
We notice in Figure \ref{figure:mse_MUSIC_GMUSIC_wide_uncorr} that the performance of G-MUSIC, MUSIC as well as the DoA estimate from the spatial periodogram coincide, since the source DoA are widely spaced (five times the beamwidth $\frac{2 \pi}{M}$). We also notice that the threshold effect of the spatial periodogram is less significant, since it is not constrained by the subspace separation condition (see Section \ref{section:spatial_periodogram}).

In Figure \ref{figure:mse_MUSIC_GMUSIC_wide_corr}, we consider the same simulation as for Figure \ref{figure:mse_MUSIC_GMUSIC_wide_uncorr}, except that we add
significant correlation between sources, by taking $\S_N=\R^{1/2}\X_N$ with $\R=[1,0.4;0.4,1]$ and $\X_N$ having standard i.i.d $\Ncal_{\Cbb}(0,1)$ entries. 
\begin{figure}[h!]
	\centering
	\includegraphics[scale=0.4]{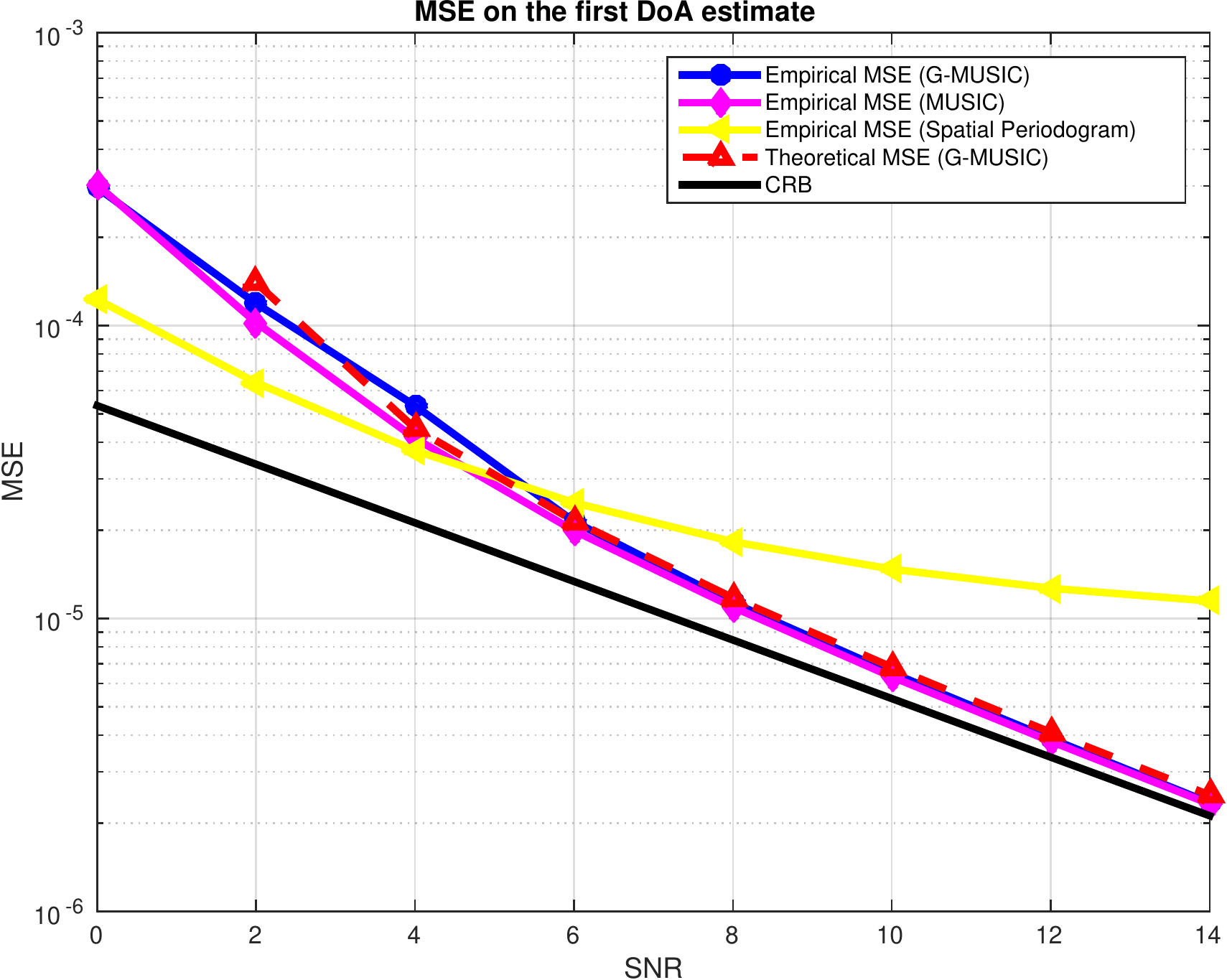}
	\caption
	{
		Empirical MSE of $\hat{\theta}_{1,N}$, $\hat{\theta}^{(t)}_{1,N}$ and $\hat{\theta}^{(p)}_{1,N}$ for widely spaced DoA and correlated source signals,
		versus SNR (dB)
	}
	\label{figure:mse_MUSIC_GMUSIC_wide_corr}
\end{figure}
Again, we notice that both G-MUSIC and MUSIC perform well, since the DoA are widely spaced. Concerning the spatial periodogram method, we notice that a strong bias occurs at high SNR, which corresponds to the well-known effect of source correlation on spatial beamforming techniques (see \cite{reddy1987performance}).

To illustrate the asymptotic Gaussianity of the G-MUSIC and MUSIC estimates predicted in Theorems \ref{theorem:clt_GMUSIC} and \ref{theorem:clt_MUSIC}, we plot
in Figure \ref{figure:hist} the histograms of $\hat{\theta}_{2,N}$ and $\hat{\theta}_{2,N}^{(t)}$ (5000 draws), with the parameters used in Figure   \ref{figure:mse_MUSIC_GMUSIC_wide_corr} (widely spaced DoA and correlated source signals, SNR=$6$ dB).

\begin{figure}
	\centering
	\subfigure[G-MUSIC]{\includegraphics[scale=0.4]{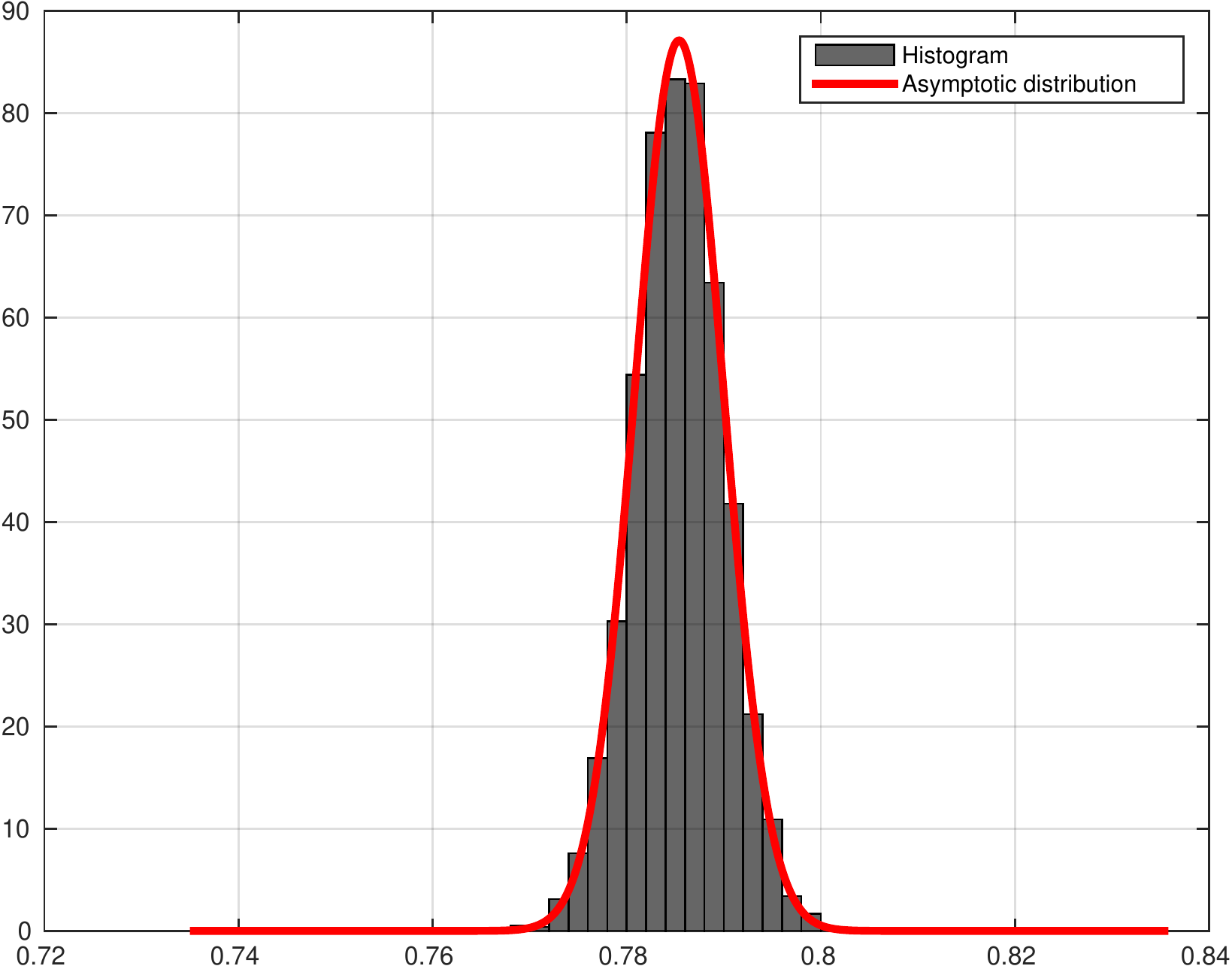}}
	\subfigure[MUSIC]{\includegraphics[scale=0.4]{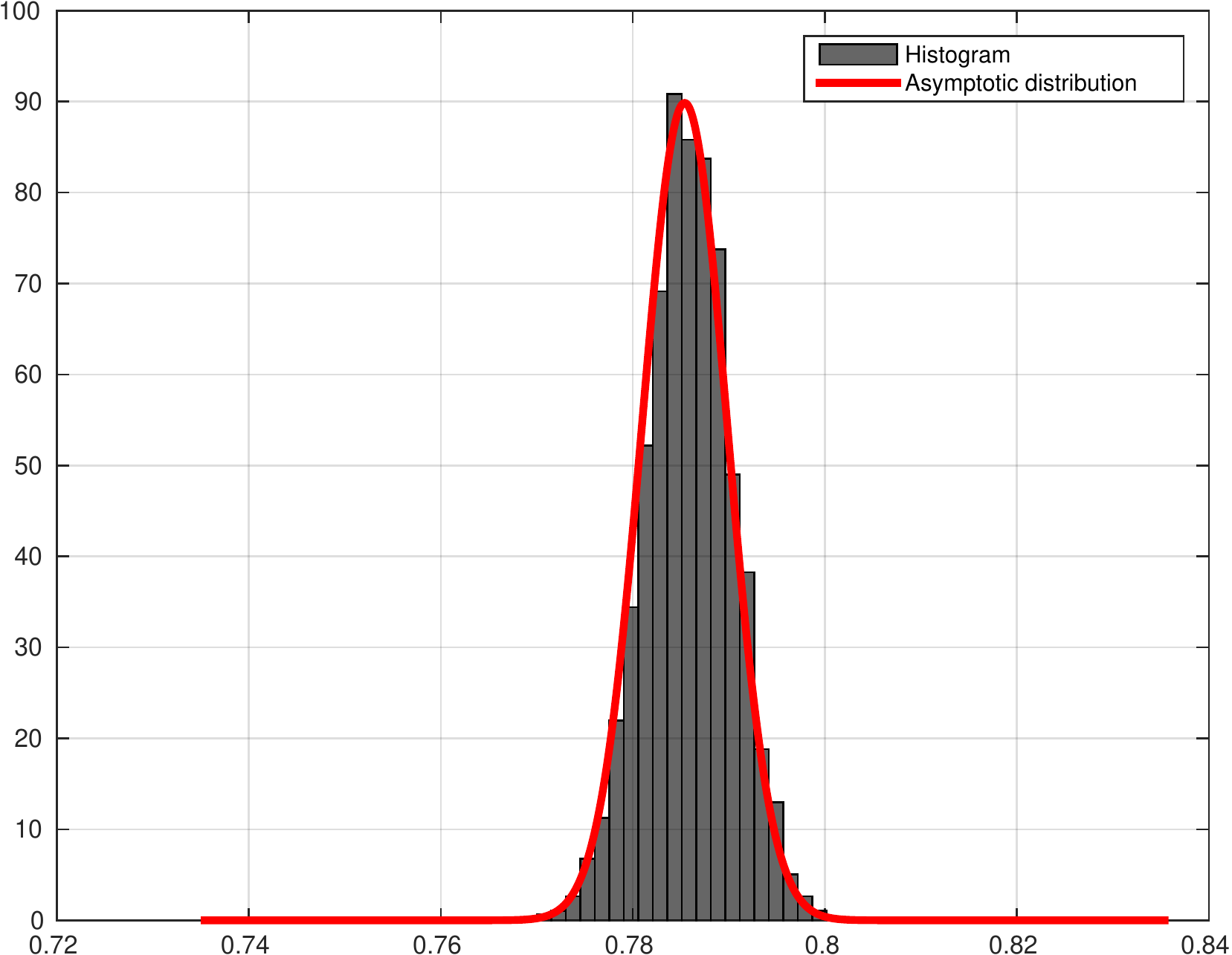}}
	\caption{Histograms of $\hat{\theta}_{2,N}$ and $\hat{\theta}_{2,N}^{(t)}$ compared to their asymptotic Gaussian distribution}
	\label{figure:hist}
\end{figure}

Figure \ref{figure:mse_MUSIC_GMUSIC_close} illustrates the closely spaced DoA scenario, and the parameters are the same as in Figure \ref{figure:mse_MUSIC_GMUSIC_wide_uncorr}, except for the DoA fixed to $\theta_1=0$, $\theta_2= 0.25 \times \frac{2\pi}{M}$.
The separation condition is fulfilled for all SNR. 
\begin{figure}
	\centering
	\includegraphics[scale=0.4]{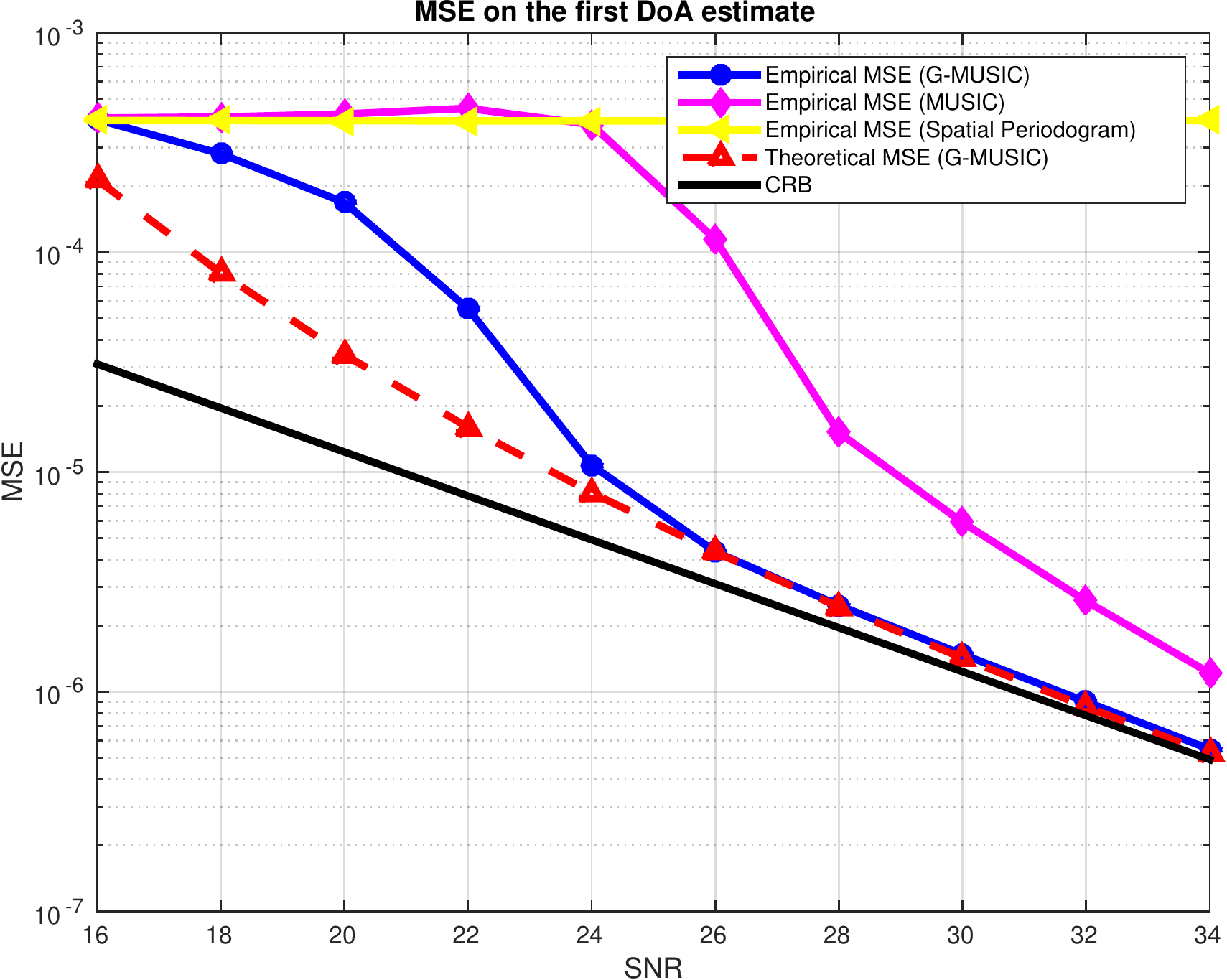}
	\caption{Empirical MSE of $\hat{\theta}_{1,N}$, $\hat{\theta}^{(t)}_{1,N}$ and $\hat{\theta}^{(p)}_{1,N}$ for closely spaced DoA versus SNR (dB)}
	\label{figure:mse_MUSIC_GMUSIC_close}
\end{figure}
One can observe that a strong difference occurs between the performances of the G-MUSIC and MUSIC methods, e.g. a difference of 4 dB between the threshold points 
of G-MUSIC and MUSIC can be measured, which illustrates the result of Theorem \ref{theorem:consistency_close}. Moreover, we notice the poor performance
of the spatial periodogram DoA estimate, which suffers from the well-known resolution loss, since the DoA spacing is lower than a beamwidth.

Similarly, in Figure \ref{figure:mse_MUSIC_GMUSIC_close_undersampled}, we keep the same parameters as for Figure \ref{figure:mse_MUSIC_GMUSIC_close}
except that $M=40$ and $N=20$. Thus, we consider an ``undersampled" scenario in which $N > M$.
\begin{figure}
	\centering
	\includegraphics[scale=0.4]{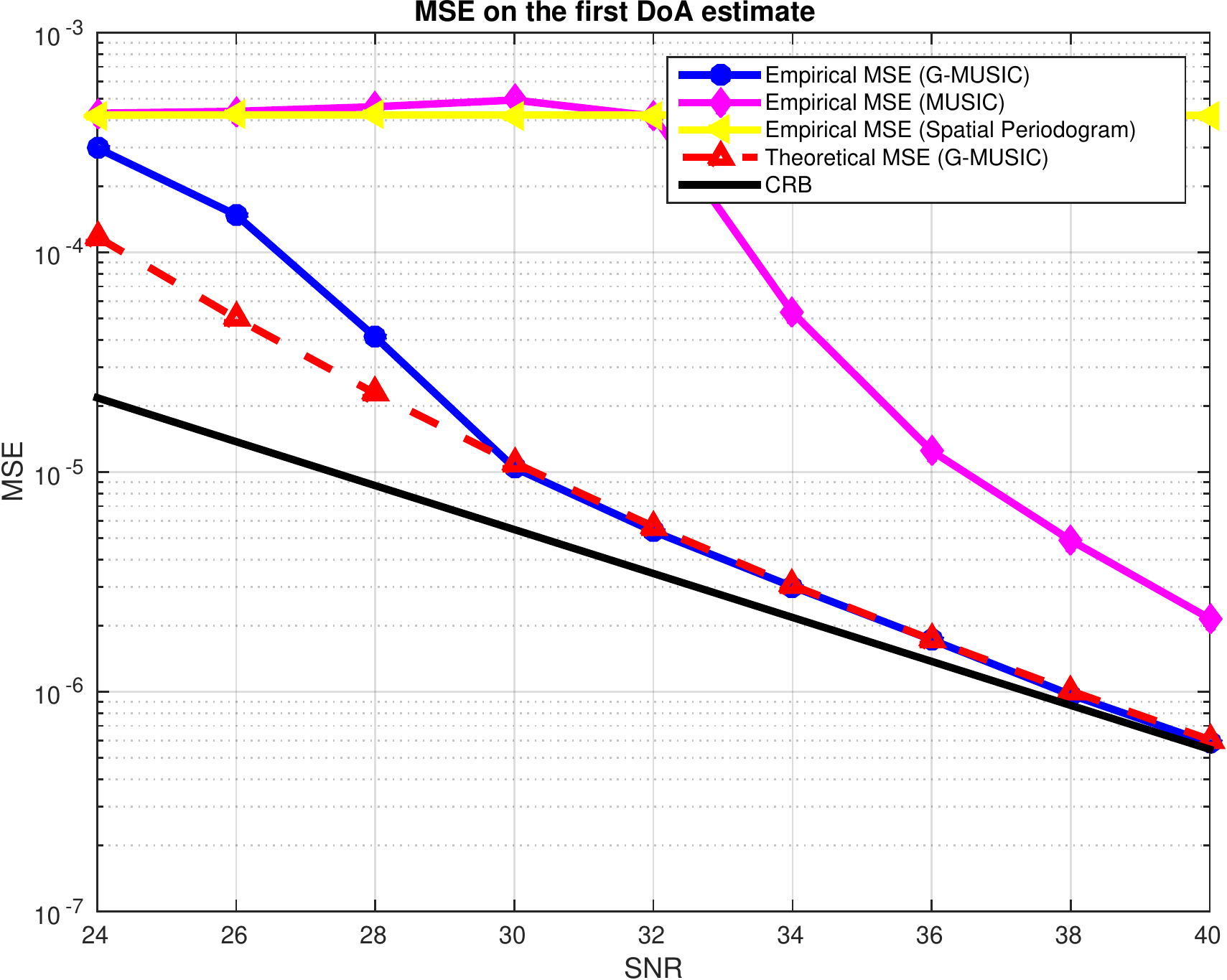}
	\caption
	{
		Empirical MSE of $\hat{\theta}_{1,N}$, $\hat{\theta}^{(t)}_{1,N}$ and $\hat{\theta}^{(p)}_{1,N}$ for closely spaced DoA (undersampled), 
		versus SNR (dB)
	}
	\label{figure:mse_MUSIC_GMUSIC_close_undersampled}
\end{figure}
In that case, G-MUSIC still outperforms the MUSIC estimates, with about 6 dB between the threshold points.

In Figure \ref{figure:mse_MUSIC_wide}, we provide the empirical MSE of MUSIC together with the theoretical MSE given in Theorem \ref{theorem:clt_MUSIC}.
The parameters are $M=40$, $N=80$, $\theta_1=0$, $\theta_2 = 5 \times \frac{2 \pi}{M}$, and correlated source signals with correlation matrix $\R=[1,0.4;0.4,1]$ and the separation condition occurs near $2$ dB.
\begin{figure}
	\centering
	\includegraphics[scale=0.4]{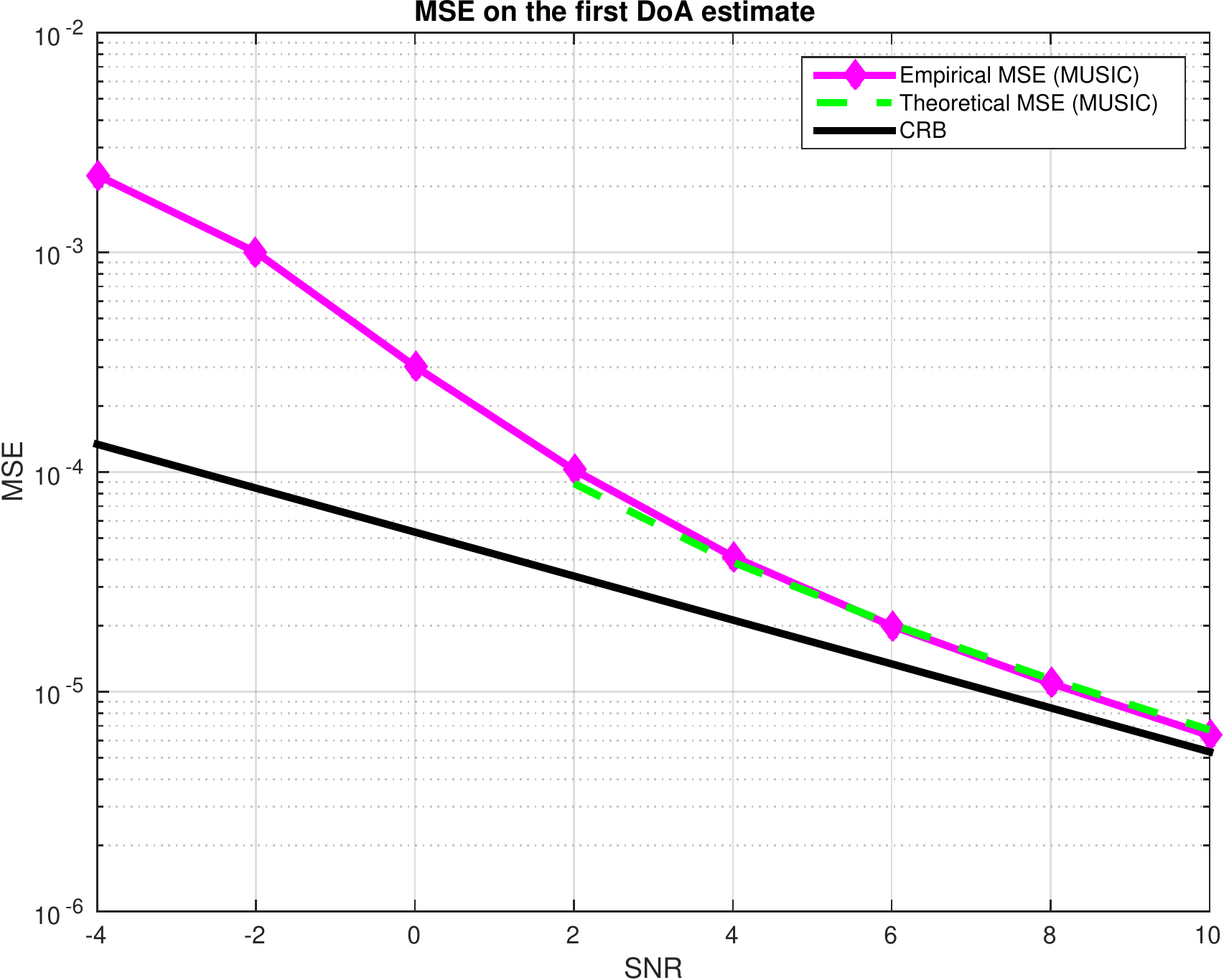}
	\caption{Empirical MSE of $\hat{\theta}_{1,N}^{(t)}$ for widely spaced DoA and correlated sources, versus SNR (dB)}
	\label{figure:mse_MUSIC_wide}
\end{figure}
One can observe the accuracy of the theoretical MSE predicted in Theorem \ref{theorem:clt_MUSIC}.	

Finally, we provide in Figures \ref{figure:mse_GMUSIC_cond_uncond_wide} and \ref{figure:mse_GMUSIC_cond_uncond_close} a comparison between the conditional and unconditional G-MUSIC estimates, using respectively the same scenarios as for Figure \ref{figure:mse_MUSIC_GMUSIC_wide_uncorr} and \ref{figure:mse_MUSIC_GMUSIC_close}.
The unconditional G-MUSIC estimator is computed with the formula of \cite{mestre2008modified}.
We observe that the two estimators exhibit the same empirical MSE as soon as the separation condition is fulfilled (around SNR = 2 dB for Figure \ref{figure:mse_GMUSIC_cond_uncond_wide} and verified for all SNR in Figure \ref{figure:mse_GMUSIC_cond_uncond_close}), which illustrates the remarks in Section \ref{section:connection} on the connections between both estimators.
	
\begin{figure}
	\centering
	\includegraphics[scale=0.4]{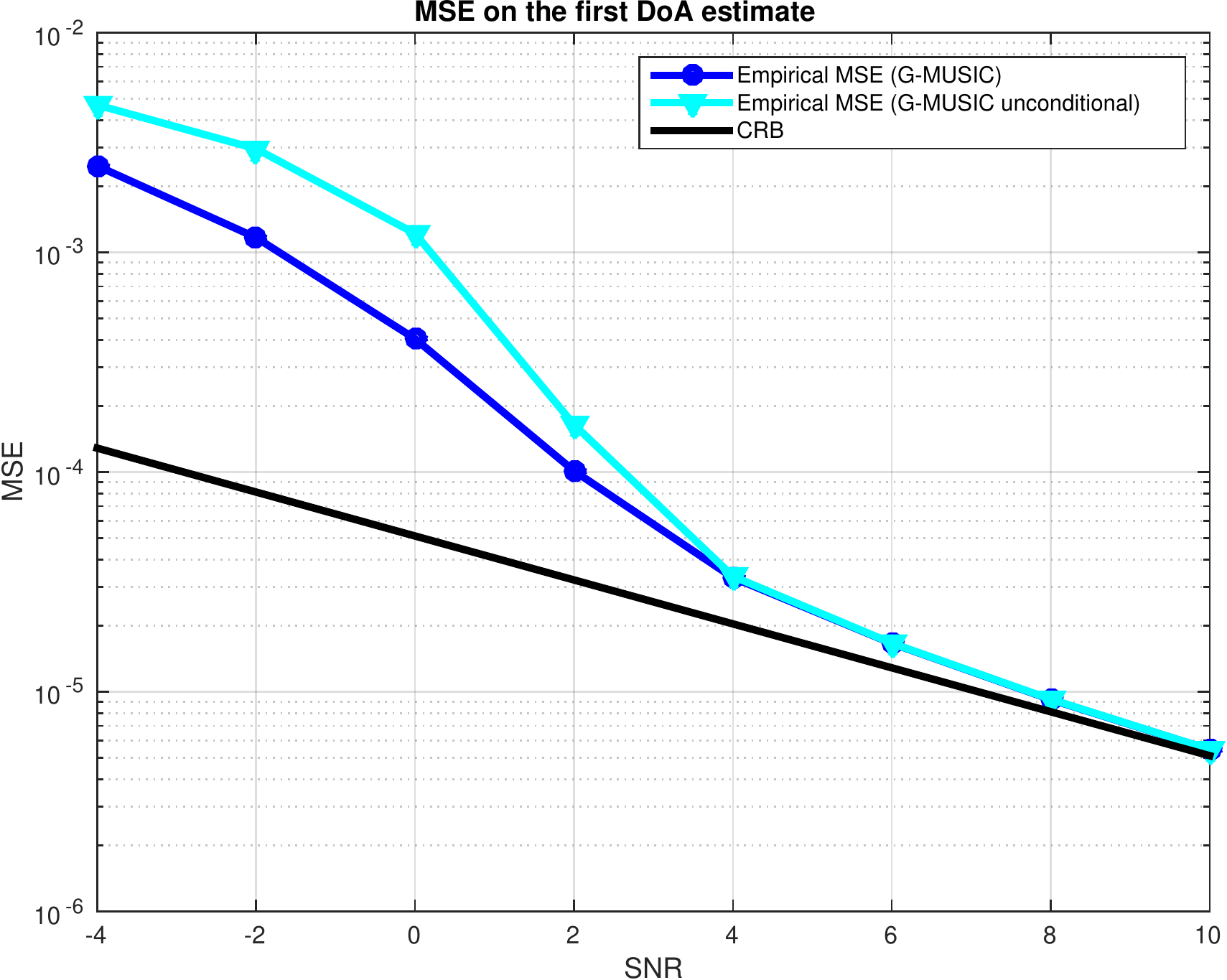}
	\caption{Empirical MSE of $\hat{\theta}_{1,N}$ (conditional and unconditional G-MUSIC), for widely spaced DoA and uncorrelated source, versus SNR (dB)}
	\label{figure:mse_GMUSIC_cond_uncond_wide}
\end{figure}
	
\begin{figure}
	\centering
	\includegraphics[scale=0.4]{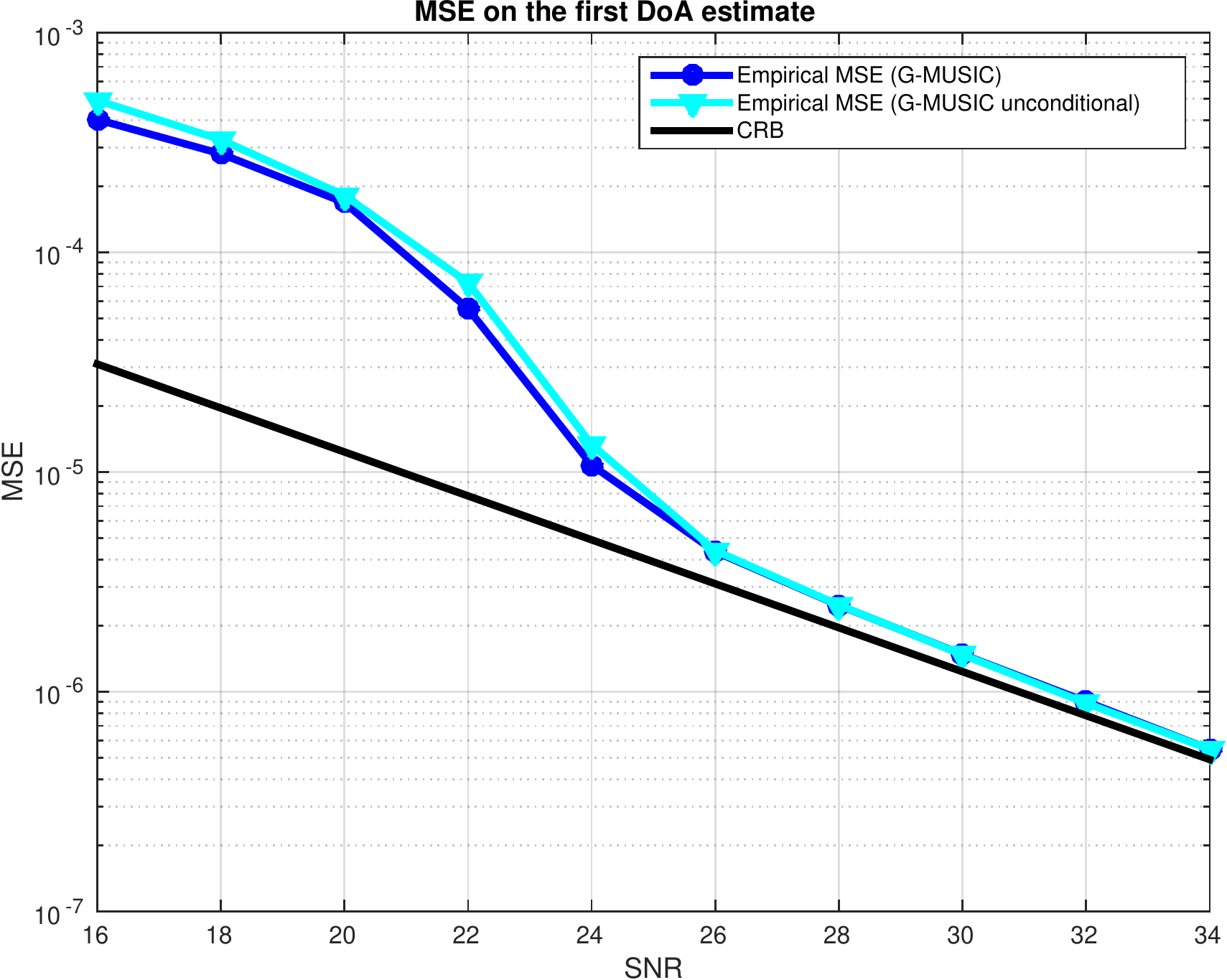}
	\caption{Empirical MSE of $\hat{\theta}_{1,N}$ (conditional and unconditional G-MUSIC), for closely spaced DoA and uncorrelated sources, versus SNR (dB)}
	\label{figure:mse_GMUSIC_cond_uncond_close}
\end{figure}

	\section{Conclusion}
	
In this paper, we have adressed a statistical comparison of the performance of the G-MUSIC and MUSIC method for DoA estimation, in an asymptotic regime where 
the number of sensors $M$ and the number of samples $N$ both converge to infinity at the same rate. Two scenarios were considered.
In a first scenario where the source DoA are widely spaced (i.e. fixed with respect to $M$,$N$), we have proved that both MUSIC and G-MUSIC exhibit the same asymptotic performance 
in terms of consistency and asymptotic Gaussianity, 
%and that the use of low resolution techniques such that a spatial periodogram provides consistent estimates, 
%which confirms the well known fact that the use of high resolution techniques is not necessary for these situations.
In a second scenario where the source DoA are closely spaced (i.e. with an angular separation of the order of a beamwidth $\Ocal(M^{-1})$), we have proved that G-MUSIC is still $N$-consistent, while MUSIC is no more able to separate the DoA. The asymptotic Gaussianity of G-MUSIC and the identification of its asymptotic MSE provided in this paper hold under general conditions,
including correlation between sources, and extend previous existing results which were only valid for asymptotically uncorrelated source signals.

\appendices

\section{Comparison between the unconditional subspace estimator of \cite{mestre2008modified} with the estimator (\ref{eq:noise_subspace_estimator}).}
\label{appendix:connections-mestre}

In this section, we establish (\ref{eq:cv-spike-unconditional}) when the source signals
are deterministic signals satisfying Assumption \ref{assumption:spiked}. For this, 
we first recall that the unconditional estimator $\hat{\eta}_{N,u}$ proposed in 
(\ref{eq:noise_subspace_estimator}) can be written as
\begin{align}	
	\hat{\eta}_{N,u} 
	= 
	\frac{1}{2 i \pi} \int_{\partial \Rcal} \d_{1,N}\left( \Sigmabs_N \Sigmabs_N^{*} - z {\bf I} \right)^{-1} \d_{2,N} \ \hat{g}_N(z) \drm z,
	\label{eq:expre-eta-unconditional}
\end{align}
where $\hat{g}_N(z)$ is defined by
\begin{align}
	\hat{g}_N(z) = \frac{(1 - c_N) + c z^{2} \hat{m}_N^{'}(z)}{(1 - c_N) - c_N z \hat{m}_N(z)},
	\notag
\end{align}
and where $\partial \Rcal$ is a contour enclosing the interval $[\sigma^{2}(1 - \sqrt{c})^{2} - \epsilon, \sigma^{2}(1 + \sqrt{c})^{2} + \epsilon]$, $\epsilon$ being chosen in such a way that $\sigma^{2}(1 + \sqrt{c})^{2} + \epsilon < \lambda_K$, and  
where we recall that $\hat{m}_N(z) = \frac{1}{M} \mathrm{tr} \left( \Sigmabs_N \Sigmabs_N^{*} - z {\bf I} \right)^{-1}$ ($\hat{m}_N^{'}(z)$ represents the derivative of $\hat{m}_N(z)$ w.r.t. $z$). 
Using condition (\ref{eq:condition-cN}), it is easily seen that $\hat{m}_N(z) = m(z) + o_{\Pbb}\left(\frac{1}{\sqrt{N}}\right)$, that $\hat{m}^{'}_N(z) = m^{'}(z) + o_{\Pbb}\left(\frac{1}{\sqrt{N}}\right)$, and using an additional argument such as in \cite[Sec. 4.1]{arxiv_clt}, one can show
\begin{align}
	&\hat{\eta}_{N,c} 
	 =
	 \notag\\  
	&\frac{1}{2 i \pi} \int_{\partial \Rcal} \d_{1,N}^* \left( \Sigmabs_N \Sigmabs_N^{*} - z {\bf I} \right)^{-1} \d_{2,N} 
	\ g(z) \drm z 
	+ o_{\Pbb}\left(\frac{1}{\sqrt{N}}\right),
	\notag
\end{align}
where 
\begin{align}
	g(z) = 	\frac{(1 - c) + c z^{2} m^{'}(z)}{(1 - c) - c z m(z)}.
  	\notag
\end{align}
It is established in \cite{arxiv_clt} that 
\begin{align}
	&\hat{\eta}_N = 
	\notag\\	
	&\quad\frac{1}{2 i \pi} \int_{\partial \Rcal} 
	\d_{1,N}^* \left( \Sigmabs_N \Sigmabs_N^{*} - z {\bf I} \right)^{-1} \d_{2,N} \frac{w^{'}(z)}{1 + \sigma^{2} c m(z)}  \drm z
	\notag\\	
	&\quad + o_{\Pbb}\left(\frac{1}{\sqrt{N}}\right).
	\notag
\end{align}
The conclusion follows from the identity
\begin{align}
	g(z) = \frac{w^{'}(z)}{1 + \sigma^{2} c m(z)}.
	\notag
\end{align}
which can be checked easily.

\section{Proof of Theorem \ref{theorem:consistency_MUSIC_GMUSIC_low}}
\label{appendix:proof_consistency_low}

The consistency of G-MUSIC is already established in \cite{hachem2012large}, and we prove hereafter the consistency of MUSIC.

From Theorem \ref{theorem:spiked_eigv}, we have for all $\theta \in [-\pi,\pi]$,
\begin{align}
	\hat{\eta}_N^{(t)}(\theta) = 	\eta_N^{(t)}(\theta) + o(1),
	\label{eq:conv_music}
\end{align}
with probability one, where 
\begin{align}
	\eta_N^{(t)}(\theta) = 1 - \a(\theta)^*\U_N \D \U_N^* \a(\theta),
	\notag
\end{align}
with $\U_N = [\u_{1,N},\ldots,\u_{K,N}]$ and $\D = \diag(d_1,\ldots,d_K)$ with
\begin{align}
	d_{k} = \frac{\lambda_{k}^2-\sigma^4 c}{\lambda_{k} \left(\lambda_{k}+\sigma^2 c\right)}.
	\notag
\end{align}
It is easily seen that $d_1 > d_2 > \ldots > d_K$. 
Applying verbatim the steps of \cite[Sec. 3.3]{hachem2012large}, \eqref{eq:conv_music} can be strengthened to
\begin{align}
	\sup_{\theta \in [-\pi,\pi]} \left|\hat{\eta}_N^{(t)}(\theta) - \eta_N^{(t)}(\theta)\right| \xrightarrow[N\to\infty]{a.s.} 0.
	\label{eq:uniform_conv_MUSIC_loc}
\end{align}
Using the fact that $\U_N$ and $\A$ share the same image, we have
\begin{align}
	\U_N = \A \left(\A^*\A\right)^{-1/2} \V_N,
	\notag
\end{align}
where $\V_N$ is a $K \times K$ unitary matrix given by $\V_N = (\A^*\A)^{-1/2}\A^*\U_N$. 
Since $\theta_1,\ldots,\theta_K$ are fixed with respect to $N$, we also have $\A^*\A \to \I_K$ as $N \to \infty$.
It is clear that if $l \neq k$, then it holds that 
$$
\sup_{\theta \in \Ical_k} |\a(\theta)^*\a(\theta_l)| \rightarrow 0
$$ 
From this, we obtain immediately that 
$$
\sup_{\theta \in \Ical_k} \| \a(\theta)^* \U_N  - \a(\theta)^* \a(\theta_k) \e_k^* \V_N \| \rightarrow 0
$$
and that, for all $k=1,\ldots,K$,
\begin{align}
	\sup_{\theta \in \Ical_k} 
	\left|
		\eta_N^{(t)}(\theta) 
		- \left(1- \left|\a(\theta)^*\a(\theta_k)\right|^2 \e_k^* \V_N\D\V_N^* \e_k\right)
	\right|
	\xrightarrow[N \to \infty]{} 0.
	\label{eq:formula_asymp_MUSIC_loc}
\end{align}
Moreover, it holds that 
\begin{align}
	\sup_{\theta \not\in \bigcup_k \Ical_k} 
	\eta_N^{(t)}(\theta) \xrightarrow[N\to\infty]{} 1.
	\label{eq:formula_asymp_MUSIC_loc2}
\end{align}
We claim that
\begin{align}
	\hat{\theta}_{k,N}^{(t)} \xrightarrow[N \to \infty]{a.s.} \theta_k.
	\label{eq:simple_consistency_MUSIC}
\end{align} 
To verify this, we first remark that (\ref{eq:uniform_conv_MUSIC_loc}) and
(\ref{eq:formula_asymp_MUSIC_loc}) 
used at point $\theta = \hat{\theta}_{k,N}^{(t)}$ lead to
$$
\hat{\eta}_N^{(t)}(\hat{\theta}_{k,N}^{(t)}) -  \left(1- \left|\a(\hat{\theta}_{k,N}^{(t)})^*\a(\theta_k)\right|^2 \e_k^* \V_N\D\V_N^* \e_k\right) \rightarrow 0
$$
almost surely. As function of $\theta$, $\left|\a(\theta)^*\a(\theta_k)\right|^2$ has a unique global maximum at $\theta_k$ and that $\e_k^* \V_N\D\V_N^* \e_k$ is lower bounded by $d_K > 0$ so we deduce that (\ref{eq:simple_consistency_MUSIC}) holds. Otherwise, one could extract 
from sequence $(\hat{\theta}_{k,N}^{(t)})_{N \geq 1}$ a 
subsequence $\hat{\theta}_{k,\phi(N)}^{(t)}$ converging towards a point $\theta_* \neq \theta_k$ almost surely. 
This would imply that 
$$
\hat{\eta}_{\phi(N)}^{(t)}(\hat{\theta}_{k,\phi(N)}^{(t)}) - \hat{\eta}_{\phi(N)}^{(t)}(\theta_*) \rightarrow 0
$$
and that 
$$
\hat{\eta}_{\phi(N)}^{(t)}(\hat{\theta}_{k,\phi(N)}^{(t)}) -  \left(1- \left|\a(\theta_*)^*\a(\theta_k)\right|^2 \e_k^* \V_N\D\V_N^* \e_k\right) \rightarrow 0
$$
However, (\ref{eq:uniform_conv_MUSIC_loc}) and
(\ref{eq:formula_asymp_MUSIC_loc}) used at point $\theta = \theta_k$ imply that
$$
\hat{\eta}_N^{(t)}(\theta_k) -  \left(1-  \e_k^* \V_N\D\V_N^* \e_k\right) \rightarrow 0.
$$
Therefore, for $\epsilon > 0$ small enough, it holds that 
$$
\hat{\eta}_{\phi(N)}^{(t)}(\theta_k) < \hat{\eta}_{\phi(N)}^{(t)}(\hat{\theta}_{k,\phi(N)}^{(t)}) - \epsilon
$$
for each $N$ large enough, a contradiction.

We now improve \eqref{eq:simple_consistency_MUSIC} by showing that
\begin{align}
	N\left(\hat{\theta}_{k,N}^{(t)} - \theta_k\right) \xrightarrow[N \to \infty]{a.s.} 0,
	\label{eq:strong_consistency_MUSIC}
\end{align}
and for that purpose we follow the approach of \cite{hannan1973estimation} (also used in \cite[Sec. 4]{hachem2012large}).
By definition, we have
\begin{align}
	\eta^{(t)}_N(\hat{\theta}_{k,N}^{(t)})
	&\leq 
	\left|\eta^{(t)}_N(\hat{\theta}_{k,N}^{(t)}) - \hat{\eta}^{(t)}_N(\hat{\theta}_{k,N}^{(t)})\right| + \hat{\eta}^{(t)}_N(\hat{\theta}_{k,N}^{(t)})
	\notag\\
	&\leq
	\sup_{\theta\in[-\pi,\pi]}\left|\eta^{(t)}_N(\theta) - \hat{\eta}^{(t)}_N(\theta)\right| + \hat{\eta}^{(t)}_N(\theta_k),
	\notag
\end{align}
and from \eqref{eq:uniform_conv_MUSIC_loc} and (\ref{eq:formula_asymp_MUSIC_loc}) used at point $\theta = \theta_k$, we obtain
\begin{align}
	\limsup_{N \to \infty} \eta^{(t)}_N(\hat{\theta}_{k,N}^{(t)})
	&\leq 
	\limsup_{N \to \infty} \hat{\eta}_N^{(t)}(\theta_k)
	\notag\\
	&= 1 - \liminf_{N \to \infty} \e_k^* \V_N\D\V_N^* \e_k
	\notag\\
	&< 1,
	\label{eq:MUSIC_loc_bound}
\end{align}
where the last inequality comes from the fact that $\e_k^* \V_N\D\V_N^* \e_k \geq d_K > 0$. 
Assume that the sequence $N\left(\hat{\theta}_{k,N}^{(t)} - \theta_k\right)$ is not bounded. Then we can extract a subsequence 
$\varphi(N)\left(\hat{\theta}_{k,\varphi(N)}^{(t)} - \theta_k\right)$ such that 
\begin{align}
	\varphi(N)\left|\hat{\theta}_{k,\varphi(N)}^{(t)} - \theta_k\right| \xrightarrow[N \to \infty]{} \infty.
	\notag
\end{align}
This implies that $\a(\hat{\theta}_{k,\phi(N)}^{(t)})^*\a(\theta_k) \rightarrow 0$ and that, by 
(\ref{eq:formula_asymp_MUSIC_loc}),  
$\eta_{\varphi(N)}^{(t)}(\hat{\theta}_{k,\varphi(N)}^{(t)}) \to 1$, a contradiction with \eqref{eq:MUSIC_loc_bound}.
Since $N\left(\hat{\theta}_{k,N}^{(t)} - \theta_k\right)$ is bounded, we can extract a subsequence such that
\begin{align}
	\varphi(N)\left|\hat{\theta}_{k,\varphi(N)}^{(t)} - \theta_k\right| 
	\xrightarrow[N \to \infty]{} \beta,
	\notag
\end{align}
with $\beta$ assumed to lie in $[-\pi,\pi]$ without loss of generality. 
If $\beta \neq 0$, then \eqref{eq:formula_asymp_MUSIC_loc} gives
\begin{align}
	\eta_{\varphi(N)}^{(t)}(\hat{\theta}_{k,\varphi(N)}^{(t)}) 
	=
	1 - \e_k^* \V_N\D\V_N^* \e_k \sinc\left(\beta c/2\right) + o(1) 
% 	\xrightarrow[N \to \infty]{} 1-\frac{\alpha_k^2-\sigma^4 c}{\alpha_k \left(\alpha_k+\sigma^2 c\right)} \sinc\left(\beta/2\right)  
% 	> 1 - \frac{\alpha_k^2-\sigma^4 c}{\alpha_k \left(\alpha_k+\sigma^2 c\right)},
	\notag
\end{align}
with probability one.
Since, in that case,
\begin{align}
	\limsup_{N \to \infty} \eta^{(t)}_N(\hat{\theta}_{k,N}^{(t)})
	> 1 - \liminf_{N\to\infty} \e_k^* \V_N\D\V_N^* \e_k,
	\notag
\end{align}
this contradicts \eqref{eq:MUSIC_loc_bound} again.

Therefore all converging subsequences of the bounded sequence $N\left(\hat{\theta}_{k,N}^{(t)} - \theta_k\right)$ have the same limit, which is $0$, 
and thus the whole sequence converges itself to $0$, which finally shows \eqref{eq:strong_consistency_MUSIC}.

\section{Proof of Theorem \ref{theorem:consistency_close}}
\label{appendix:proof_consistency_close}

Recall from \eqref{eq:uniform_consistency} that we have $\sup_{\theta}\left|\hat{\eta}_N(\theta) - \eta_N(\theta)\right|\to_N 0$ with probability one, with 
\begin{align}
	\eta_N(\theta)=\a(\theta)^*\Pibs_N\a(\theta) = 1 - \a(\theta)^*\A\left(\A^*\A\right)^{-1}\A^*\a(\theta).
	\notag
\end{align}
From Assumption \ref{assumption:close_DoA}, we have $\sup_{\theta}\left|\eta_N(\theta) - \tilde{\eta}_N(\theta)\right|\to_N 0$ where
\begin{align}
	\tilde{\eta}_N(\theta) = 
	1 - \frac{1}{1-\sinc\left(\frac{\alpha c}{2}\right)^2}\a(\theta)^* \A \T \A^* \a(\theta),
	\label{eq:eta_tilde}
\end{align}
where
\begin{align}
	\T = 
	\begin{bmatrix}
		1 & -\erm^{\irm \alpha c/2} \sinc\left(\frac{\alpha c}{2}\right)
		\\
		-\erm^{-\irm \alpha c/2} \sinc\left(\frac{\alpha c}{2}\right) & 1
	\end{bmatrix}
	.
	\notag
\end{align}
Note that
\begin{align}
	\limsup_{N\to\infty} \left|\hat{\eta}_N(\hat{\theta}_{1,N})\right| \leq \limsup_{N\to\infty} \left|\hat{\eta}_N(\theta_{1,N})\right| = 0.
	\label{eq:MUSIC_loc_bound2}
\end{align}
We next rely on the following lemma.
\begin{lemma}
	\label{lemma:kappa}
	If $(\psi_N)$ is a sequence of $[-\pi,\pi]$ such that $N\left|\psi_N - \theta_{1,N}\right| \to \infty$, then
	\begin{align}
		\eta_N(\psi_N) \xrightarrow[N\to\infty]{} 1.
		\notag
	\end{align}
	Moreover, for any compact $\Kcal \subset \Rbb$,	
	\begin{align}
		\sup_{\beta \in \Kcal} 
		\left|\eta_N\left(\theta_{1,N} + \frac{\beta}{N}\right) - \left(1-\kappa(\beta)\right)\right|
		\xrightarrow[N\to\infty]{} 0,
		\notag
	\end{align}
	where 
	\begin{align}
		\begin{split}
			\kappa(\beta)  = 
			\frac{1}{1 - \sinc\left(\alpha c /2\right)^2}
			\Biggl(
				\sinc\left(\beta c / 2 \right)^2 
				+ \sinc\left((\beta-\alpha)c/2\right)^2 
				\notag\\
				- 2 \sinc\left(\alpha c /2\right) \sinc\left(\beta c /2\right) \sinc\left((\beta-\alpha) c /2\right)
			\Biggr)
			\notag
		\end{split}
	\end{align}
	is such that $\kappa(\beta) \leq 1$ with equality if and only if $\beta=0$ or $\beta=\alpha$.
\end{lemma}
\begin{proof}
	The two convergences can be easily obtained from \eqref{eq:eta_tilde}. 
	It thus remains to establish that $\kappa(\beta) \leq 1$ 
	with equality if and only if $\beta = 0$ or $\beta = \alpha$.
	Consider the Hilbert space $\Lcal_{\Cbb}^2\left([0,1]\right)$ endowed with the usual
	scalar product $<z_1,z_2>=\int_{0}^1 z_1(t)z_2(t)^* \drm t$, and let
	$x_1,x_2,y \in \Lcal_{\Cbb}^2\left([0,1]\right)$ defined by
	\begin{align}
		x_1(t) = 1, \quad x_2(t) = \erm^{\irm \alpha c t} 
		\quad\text{and}\quad 
		y(t) = \erm^{\irm \beta c t}.
		\notag
	\end{align}
	Straightforward computations show that $\kappa(\beta)$ coincides with the squared norm
	of the orthogonal projection of $y$ onto $\mathrm{span}\{x_1,x_2\}$. 
	Since $y$ is unit-norm, it is clear that $\kappa(\beta) \leq 1$, 
	and the equality holds if and only if $y \in \mathrm{span}\{x_1,x_2\}$, which is
	obviously the case 	if and only if $\beta=0$ or $\beta=\alpha$.
\end{proof}
From Lemma \ref{lemma:kappa}, the function $\kappa$ admits a global maximum, equal to $1$, at the unique points $0$ and $\alpha$ and
\begin{align}
		\sup_{\beta \in [-\frac{\alpha}{2}, \frac{3\alpha}{2}]} 
		\left|\hat{\eta}_N\left(\theta_{1,N} + \frac{\beta}{N}\right) - \left(1-\kappa(\beta)\right)\right|
		\xrightarrow[N\to\infty]{} 0.
		\notag
\end{align}
Thus, 
\begin{align}
	\hat{\eta}_N\left(\hat{\theta}_{1,N}\right) = 1 - \kappa\left(N(\hat{\theta}_{1,N} - \theta_{1,N})\right) + o(1),
	\notag
\end{align}
and since $\hat{\eta}_N\left(\hat{\theta}_{1,N}\right) \to 0$ a.s., we deduce that $N(\hat{\theta}_{1,N} - \theta_{1,N}) \to 0$ a.s.
We obtain similarly the same results for $\hat{\theta}_{2,N}$.

We now consider the consistency of the traditional MUSIC estimates.
From \eqref{eq:uniform_conv_MUSIC_loc},
\begin{align}
	\sup_{\theta \in [-\pi,\pi]} \left|\hat{\eta}_N^{(t)}(\theta) - \eta_N^{(t)}(\theta)\right| \xrightarrow[N \to \infty]{a.s.} 0,
	\notag
\end{align}
where $\eta_N^{(t)}(\theta) = 1- \a(\theta)^* \U_N \D \U_N^* \a(\theta)$.
From Assumption \ref{assumption:close_DoA}, $N^{-1}\S_N\S_N^{*} \to \I$, 
%we have $\left\|\A\frac{\S_N\S_N^*}{N}\A^* - \A\A^*\right\| \to 0$ as $N \to \infty$, which implies that 
% \begin{align}
% 	\sup_{\theta \in [-\pi,\pi]} \left|\a(\theta)^* \left(\U_N \D \U_N^* - \tilde{\U}_N \D \tilde{\U}_N^*\right) \a(\theta) \right| 
% 	\xrightarrow[N\to \infty]{} 0,
% 	\notag
% \end{align} 
% where $\tilde{\U}_N$ is the $M \times 2 $ matrix of eigenvectors of $\A\A^*$ associated with the two non-zero eigenvalues.
and using the fact that
\begin{align}
	\A^*\A \xrightarrow[N \to \infty]{} 
	\begin{bmatrix}
		1 & \erm^{\irm \alpha c /2} \sinc\left(\alpha c /2\right)
		\\
		\erm^{- \irm \alpha c /2} \sinc\left(\alpha c /2\right) & 1
	\end{bmatrix},
	\notag 
%	\\
% 	&\qquad\qquad=
% 	\begin{bmatrix}
% 		1 & 1
% 		\\
% 		-\erm^{- \irm \alpha c /2} & \erm^{-\irm \alpha c /2}
% 	\end{bmatrix}
% 	\begin{bmatrix}
% 		\frac{1-\sinc\left(\alpha c /2\right)}{2} & 0
% 		\\
% 		0 & \frac{1+\sinc\left(\alpha c /2\right)}{2}
% 	\end{bmatrix}
% 	\begin{bmatrix}
% 		1 &  -\erm^{\irm \alpha c /2}
% 		\\
% 		1 & \erm^{\irm \alpha c /2}
% 	\end{bmatrix}
% 	\notag
\end{align}
together with a singular value decomposition of $\A$, straightforward computations yield 
\begin{align}
	\sup_{\theta} \left|\eta_N^{(t)}(\theta) - \tilde{\eta}_N^{(t)}(\theta)\right| \xrightarrow[N \to \infty]{} 0, 
	\notag
\end{align}	
where $\tilde{\eta}_N^{(t)}(\theta) = 1 - \a(\theta)^* \A\tilde{\V} \tilde{\D} \tilde{\V}^* \A^* \a(\theta)$
and where $\tilde{\V}$ is $2 \times 2$ unitary matrix given by
\begin{align}
	\tilde{\V} =
	\frac{1}{\sqrt{2}}
	\begin{bmatrix}
		1 & \erm^{\irm \alpha c/2}
		\\
		-\erm^{-\irm \alpha c/2} & 1
	\end{bmatrix},
	\notag
\end{align}
and $\tilde{\D}$ is a $2 \times 2$ diagonal matrix defined by
\begin{align}
	\tilde{\D} =
	\begin{bmatrix}
		\frac{d_{1}(\alpha)}{1 - \sinc(\alpha c/2)} & 0
		\\
		0 & \frac{d_{2}(\alpha)}{1 + \sinc(\alpha c/2)}
	\end{bmatrix}
	\notag
\end{align}
with
\begin{align}
	d_1(\alpha) = 
	\frac{\left(1 - |\sinc(\alpha c /2)|\right)^2 - \sigma^4 c }
	{\left(1 - |\sinc(\alpha c/2)|\right) \left(1-|\sinc(\alpha c /2)| + \sigma^2 c\right)}
	\notag\\
	d_2(\alpha) = 
	\frac{\left(1 + |\sinc(\alpha c /2)|\right)^2 - \sigma^4 c }
	{\left(1 + |\sinc(\alpha c/2)|\right) \left(1 + |\sinc(\alpha c /2)| 
	+ \sigma^2 c\right)}.
	\notag
\end{align}
We now use the following result, whose proof is similar to the one of Lemma \ref{lemma:kappa}.
\begin{lemma}
	\label{lemma:kappat}
	If $(\psi_N)$ is a sequence of $[-\pi,\pi]$ such that $N\left|\psi_N - \theta_{1,N}\right| \to \infty$, then
	\begin{align}
		\eta_N^{(t)}(\psi_N) \xrightarrow[N\to\infty]{} 1.
		\notag
	\end{align}
	Moreover, for any compact $\Kcal \subset \Rbb$,	
	\begin{align}
		\sup_{\beta \in \Kcal} 
		\left|\eta_N^{(t)}\left(\theta_1 + \frac{\beta}{N}\right) - \left(1-\kappa^{(t)}(\beta)\right)\right|
		\xrightarrow[N\to\infty]{} 0,
		\notag
	\end{align}
	where 
	\begin{align}
		&\kappa^{(t)}(\beta) = 
		\notag\\
		&\left(\sinc(\beta c /2) - \sinc\left((\beta-\alpha)c/2\right)\right)^2 \frac{d_1(\alpha)}{2\left(1-|\sinc(\alpha c / 2)|\right)}
		\notag\\
		&+\left(\sinc(\beta c /2) + \sinc\left((\beta-\alpha)c/2\right)\right)^2 \frac{d_2(\alpha)}{2\left(1+|\sinc(\alpha c / 2)|\right)}.
		\notag
	\end{align}
\end{lemma}
Function $\kappa^{(t)}$ does not admit in general a local maximum at $0$ or $\alpha$. In effect, 
it is easy to find values of $\alpha$ for which $\kappa^{(t)}(0)$ and $\kappa^{(t)}(\alpha)$ are not local maxima 
of function $\kappa^{(t)}$. 
For example, if $\alpha = \frac{\pi}{c}$, we easily check that $\kappa^{(t) '}(\beta) \neq 0$
for $\beta=0$ and $\beta=\alpha$. 
	%However, there also exists values of $\alpha$ (e.g. $\alpha=\frac {2\pi}{c}$, for which $\kappa^{(t)}(\beta)$ has a local maximum at $\beta=0$ and $\beta=\alpha$

From Lemma \ref{lemma:kappat}, we have with probability one,
\begin{align}
	\hat{\eta}_N^{(t)}\left(\hat{\theta}_{1,N}^{(t)}\right) = 1 - \kappa^{(t)}\left(N\left(\hat{\theta}_{1,N}^{(t)}-\theta_{1,N}\right)\right) + o(1).
	\notag
\end{align}
Assume $N\left(\hat{\theta}_{1,N}^{(t)}-\theta_{1,N}\right) \to 0$. Then $\hat{\eta}_N^{(t)}\left(\hat{\theta}_{1,N}^{(t)}\right) \to 1 - \kappa^{(t)}(0)$.
If $0$ and $\alpha$ are not local maxima of $\kappa^{(t)}$, let $\beta \in [-\frac{\alpha - \epsilon}{2},\frac{\alpha - \epsilon}{2}]$ such that $\kappa^{(t)}(0) < \kappa^{(t)}(\beta)$,
and $(\psi_N)$ a sequence such that $N\left(\psi_N - \theta_{1,N}\right) \to \beta$.
Then 
\begin{align}
	\limsup_{N \to \infty} \hat{\eta}_N^{(t)}\left(\hat{\theta}_{1,N}^{(t)}\right)
	&\leq
	\limsup_{N \to \infty} \hat{\eta}_N^{(t)}\left(\psi_N\right)
	\notag\\
	&= 1 - \kappa^{(t)}(\beta) 
	\notag\\
	&< 1 - \kappa^{(t)}(0),
	\notag
\end{align}
which is a contradiction.

	\section{Proof of Theorem \ref{theorem:clt_GMUSIC}}
	\label{appendix:proof_clt_GMUSIC}

To prove Theorem \ref{theorem:clt_GMUSIC}, we will use the classical $\Delta$-method, as in e.g. Hachem et al. \cite{hachem2012subspace}.
	
We consider the settings of Assumption \ref{assumption:fixed_DoA} or 
Assumption \ref{assumption:close_DoA}, and make appear the dependence in $N$ for the DoA 
in both scenarios, which we denote by $\theta_{1,N},\ldots,\theta_{K,N}$. 
Let $k=1,\ldots,K$. Using Theorem \ref{theorem:consistency_MUSIC_GMUSIC_low} under 
Assumption \ref{assumption:fixed_DoA} (respectively Theorem \ref{theorem:consistency_close}
under Assumption \ref{assumption:close_DoA}), as well as a Taylor expansion around $\theta_{k,N}$, we obtain
\begin{align}
	\hat{\eta}'_N\left(\hat{\theta}_{k,N}\right) &= 
	\hat{\eta}'_N\left(\theta_{k,N}\right) 
	+ \left(\hat{\theta}_{k,N} - \theta_{k,N}\right)
	\hat{\eta}^{(2)}_N\left(\theta_{k,N}\right) 
	\notag\\
	&\qquad
	+ \frac{\left(\hat{\theta}_{k,N} - \theta_{k,N}\right)^2}{2}
	\hat{\eta}^{(3)}_N\left(\tilde{\theta}_{k,N}\right),
	\notag
\end{align}
where $\tilde{\theta}_{k,N} \in \left(\min\left\{\hat{\theta}_{k,N},\theta_{k,N}\right\}, \max\left\{\hat{\theta}_{k,N},\theta_{k,N}\right\} \right)$.
Since by definition, $\hat{\eta}'_N\left(\hat{\theta}_{k,N}\right)=0$, we obtain
\begin{align}
	\hat{\theta}_{k,N} - \theta_{k,N} = 
	-\frac{\hat{\eta}'_N\left(\theta_{k,N}\right)}
	{	
		\hat{\eta}^{(2)}_N\left(\theta_{k,N}\right) 
		+ \frac{\hat{\theta}_{k,N} - \theta_{k,N}}{2}
		\hat{\eta}^{(3)}_N\left(\tilde{\theta}_{k,N}\right)
	}.
	\notag
\end{align}
As the $j$-th derivative $\a^{(j)}(\theta)$ satisfies $\sup_{\theta} \left\|\a^{(j)}(\theta)\right\| \sim M^{j}$, we deduce from \cite{hachem2012large} that 
\footnote
{
	The boundedness \eqref{eq:third_derivative} can be obtained using the techniques
	developed in the proof of \cite[Th. 3.1]{hachem2012large} (see also equation (1.3)
	in the introduction part of this reference).
}
\begin{align}
	 \frac{1}{N^3} \hat{\eta}^{(3)}_N\left(\tilde{\theta}_{k,N}\right) = \Ocal(1) 
	 \label{eq:third_derivative}
\end{align}	 
with probability one. Theorem \ref{theorem:consistency_MUSIC_GMUSIC_low} implies
\begin{align}
	\left(\hat{\theta}_{k,N} - \theta_{k,N}\right) \frac{\hat{\eta}^{(3)}_N \left(\tilde{\theta}_{k,N}\right)}{N^2} \xrightarrow[N\to\infty]{a.s.} 0,
	\notag
\end{align}
and we obtain 
\begin{align}
	N^{3/2} \left(\hat{\theta}_{k,N} - \theta_{k,N}\right) = 
	-\frac{\frac{1}{\sqrt{N}}\hat{\eta}'_N\left(\theta_{k,N}\right)}
	{	
		\frac{1}{N^2}\hat{\eta}^{(2)}_N\left(\theta_{k,N}\right) + o_{\Pbb}(1)
	}.
	\label{eq:clt_angles_temp}
\end{align}
By using \eqref{eq:consistency_subspace_estimator} and the fact that $\Pibs_N \a(\theta_k)=\mathbf{0}$, we can write
\begin{align}
		\frac{1}{N^2}\hat{\eta}^{(2)}_N\left(\theta_{k,N}\right) 
		&= 
		2\frac{\a'(\theta_{k,N})^*}{N} \Pibs_N \frac{\a'(\theta_{k,N})}{N} + o_{\Pbb}(1)
		\notag
\end{align}
Under Assumption \ref{assumption:fixed_DoA}, the basic convergences $\A^*\A \to \I$, as $N\to\infty$, as well as
\begin{align}
	\left\|\frac{1}{N} \a'(\theta_{k,N})\right\|^2  
	\xrightarrow[N\to\infty]{} \frac{c^2}{3}
	\notag
\end{align}
and
\begin{align}
	\left|\frac{1}{N} \a'(\theta_{k,N})^* \a(\theta_{\ell,N}) \right| 
	\xrightarrow[N\to\infty]{} \frac{c^2}{4} \delta_{k,\ell}
	\notag
\end{align}
prove that 
\begin{align}
	\frac{\a'(\theta_{k,N})^*}{N} \Pibs_N \frac{\a'(\theta_{k,N})}{N} 
	\xrightarrow[N\to\infty]{} \frac{c^2}{12} > 0.
	\notag
\end{align}
Under Assumption \ref{assumption:close_DoA}, we use the arguments in the proof of Lemma \ref{lemma:kappa}. Indeed, let $x_1,x_2,y \in \Lcal^2_{\Cbb}\left([0,1]\right)$ defined by
\begin{align}
	x_1(t) = 1, x_2(t) = \erm^{\irm \alpha c t} 
	\text{ and }
	y(t) = \irm c t.
	\notag
\end{align}
Then, we observe that  
$\frac{\a'(\theta_{k,N})^*}{N} \Pibs_N \frac{\a'(\theta_{k,N})}{N}$
converges to the squared norm of the orthogonal projection of $y$ 
onto $\mathrm{span}\{x_1,x_2\}^{\perp}$. Thus, we deduce again that 
\begin{align}
	\liminf_{N \to \infty}
	\frac{\a'(\theta_{k,N})^*}{N} \Pibs_N \frac{\a'(\theta_{k,N})}{N} 
	> 0.
	\notag
\end{align}
Consider now the quantity $\gamma_N$ introduced in \eqref{def:Gamma}, where we set
\begin{align}
	\d_{1,N} = \frac{\a'(\theta_{k,N})}{N} 
	\text{ and }
	\d_{2,N} = \a(\theta_{k,N}).
	\notag
\end{align}
Obviously,
\begin{align}
	\gamma_N 
	&\geq 
	\frac{\a'(\theta_{k,N})^*}{N} \Pibs_N \frac{\a'(\theta_{k,N})}{N}
	\sum_{\ell=1}^K 
	\frac{\sigma^2(\lambda_{\ell}+\sigma^2)}{4 (\lambda_{\ell}^2-\sigma^4 c)}
	\left|\a(\theta_{k,N})^*\u_{\ell,N}\right|^2
	\notag \\
	&\geq
	D \frac{\a'(\theta_{k,N})^*}{N} \Pibs_N \frac{\a'(\theta_{k,N})}{N},
	\notag
\end{align}
where 
$D = 
\min 
\left\{
	\frac{\sigma^2(\lambda_{\ell}+\sigma^2)}{4 (\lambda_{\ell}^2-\sigma^4 c)} :
	\ell=1,\ldots,K
\right\} > 0$. Therefore, under Assumption \ref{assumption:fixed_DoA} 
or Assumption \ref{assumption:close_DoA}, we obtain
\begin{align}
	\liminf_{N \to \infty} \gamma_N > 0.
	\notag
\end{align}
Since 
\begin{align}
	&\hat{\eta}'_N\left(\theta_{k,N}\right)=
	\notag\\
	& 2 N \Re\left(\frac{\a'(\theta_{k,N})^*}{N} \left(\I - \sum_{k=1}^K \frac{1}{h\left(\hat{\lambda}_{k,N} \right)} \hat{\u}_{k,N} \hat{\u}_{k,N}^*\right)  \a(\theta_{k,N})\right),
	\notag
\end{align}
Theorem \ref{theorem:clt_subspace} applied with $\d_{1,N} = \frac{\a'(\theta_{k,N})}{N}$, $\d_{2,N} =  \a(\theta_{k,N})$ gives
\begin{align}
	\frac{\hat{\eta}'_N\left(\theta_{k,N}\right)}{2\sqrt{N}\sqrt{\gamma_N}} \xrightarrow[N\to\infty]{\Dcal} \Ncal_{\Rbb}(0,1),
	\notag
\end{align}
Gathering this convergence with \eqref{eq:clt_angles_temp}, we eventually obtain \eqref{eq:CLT_G_MUSIC}.

\bibliographystyle{plain}
\bibliography{tsp2015} 

\begin{thebibliography}{10}

\bibitem{benaych2012singular}
Florent Benaych-Georges and Raj~Rao Nadakuditi.
\newblock The singular values and vectors of low rank perturbations of large
  rectangular random matrices.
\newblock {\em Journal of Multivariate Analysis}, 111:120--135, 2012.

\bibitem{bianchi-et-al-2011}
P.~Bianchi, M.~Debbah, M.~Maïda, and M.~Najim.
\newblock Performance of statistical tests for single source detection using
  random matrix theory.
\newblock {\em IEEE Transactions on Information Theory}, 57(4):2400--2419,
  2011.

\bibitem{couillet2015robust}
Romain Couillet.
\newblock {Robust spiked random matrices and a robust G-MUSIC estimator}.
\newblock {\em To appear in Journal of Multivariate Analysis}, 2015.

\bibitem{couillet2014robust}
Romain Couillet and Abla Kammoun.
\newblock {Robust G-MUSIC}.
\newblock In {\em Signal Processing Conference (EUSIPCO), 2014 Proceedings of
  the 22nd European}, pages 2155--2159. IEEE, 2014.

\bibitem{couillet2015random}
{Couillet, Romain and Pascal, Fr\'ed\'eric and Silverstein, Jack W.}
\newblock {The random matrix regime of Maronna's M-estimator with elliptically
  distributed samples}.
\newblock {\em Journal of Multivariate Analysis}, 139:56--78, 2015.

\bibitem{hachem2012large}
W.~Hachem, P.~Loubaton, X.~Mestre, J.~Najim, and P.~Vallet.
\newblock Large information plus noise random matrix models and consistent
  subspace estimation in large sensor networks.
\newblock {\em Random Matrices: Theory and Applications}, 1(2), 2012.

\bibitem{hachem2012subspace}
W.~Hachem, P.~Loubaton, X.~Mestre, J.~Najim, and P.~Vallet.
\newblock A subspace estimator for fixed rank perturbations of large random
  matrices.
\newblock {\em Journal of Multivariate Analysis}, 114:427--447, 2012.
\newblock arXiv:1106.1497.

\bibitem{hannan1973estimation}
E.J. Hannan.
\newblock The estimation of frequency.
\newblock {\em Journal of Applied probability}, 10(3):510--519, 1973.

\bibitem{abramovich-mestre-2008}
B.A. Jonhson, Y.I. Abramovich, and X.~Mestre.
\newblock {MUSIC, G-MUSIC, and maximum-likelihood performance breakdown }.
\newblock {\em IEEE Transactions on Signal Processing}, 56(8):3944--3958, 2008.

\bibitem{krichtman-nadler-2009}
S.~Krichtman and B.~Nadler.
\newblock Non-parametric detection of the number of signals: hypothesis testing
  and random matrix theory.
\newblock {\em IEEE Transactions on Signal Processing}, 57(10):3930--3941,
  2009.

\bibitem{loubaton2011almost}
Philippe Loubaton and Pascal Vallet.
\newblock Almost sure localization of the eigenvalues in a gaussian information
  plus noise model. application to the spiked models.
\newblock {\em Electron. J. Probab.}, 16:1934--1959, 2011.

\bibitem{marchenko1967distribution}
V.A. Marchenko and L.A. Pastur.
\newblock Distribution of eigenvalues for some sets of random matrices.
\newblock {\em Mathematics of the USSR-Sbornik}, 1:457, 1967.

\bibitem{mestre2008modified}
X.~Mestre and M.{\'A}. Lagunas.
\newblock Modified subspace algorithms for doa estimation with large arrays.
\newblock {\em IEEE Transactions on Signal Processing}, 56(2):598--614, 2008.

\bibitem{mestre2011asymptotic}
X.~Mestre, P.~Vallet, P.~Loubaton, and W.~Hachem.
\newblock Asymptotic analysis of a consistent subspace estimator for
  observations of increasing dimension.
\newblock In {\em IEEE Statistical Signal Processing Workshop (SSP)}, pages
  677--680. IEEE, 2011.

\bibitem{nadakuditi-edelman-2008}
R.R Nadakuditi and A.~Edelman.
\newblock Sample eigenvalue based detection of high-dimensional signals in
  white noise using relatively few samples.
\newblock {\em IEEE Transactions on Signal Processing}, 56(7):2625--2637, 2008.

\bibitem{reddy1987performance}
V~Umapathi Reddy, Arogyaswami Paulraj, and Thomas Kailath.
\newblock Performance analysis of the optimum beamformer in the presence of
  correlated sources and its behavior under spatial smoothing.
\newblock {\em IEEE Transactions on Acoustics, Speech and Signal Processing},
  35(7):927--936, 1987.

\bibitem{schmidt1986multiple}
R.~Schmidt.
\newblock Multiple emitter location and signal parameter estimation.
\newblock {\em IEEE Transactions on Antennas and Propagation}, 34(3):276--280,
  1986.

\bibitem{stoica1989music}
P.~Stoica and A.~Nehorai.
\newblock Music, maximum likelihood, and cramer-rao bound.
\newblock {\em IEEE Transactions on Acoustics, Speech and Signal Processing},
  37(5):720--741, 1989.

\bibitem{thomas1995probability}
John~K Thomas, Louis~L Scharf, and Donald~W Tufts.
\newblock The probability of a subspace swap in the svd.
\newblock {\em IEEE Transactions on Signal Processing}, 43(3):730--736, 1995.

\bibitem{vallet2012improved}
P.~Vallet, P.~Loubaton, and X.~Mestre.
\newblock {Improved Subspace Estimation for Multivariate Observations of High
  Dimension: The Deterministic Signal Case}.
\newblock {\em IEEE Transactions on Information Theory}, 58(2), Feb. 2012.
\newblock arXiv: 1002.3234.

\bibitem{vallet2012clt}
P.~Vallet, X.~Mestre, and P.~Loubaton.
\newblock {A CLT for the G-MUSIC DoA estimator}.
\newblock In {\em EUSIPCO 2012}, pages 2298--2302, 2012.

\bibitem{arxiv_clt}
P.~Vallet, X.~Mestre, and P.~Loubaton.
\newblock {A CLT for an improved subspace estimator with observations of
  increasing dimensions}.
\newblock 2015.
\newblock arXiv:1502.02501.

\bibitem{vinogradova2013statistical}
Julia Vinogradova, Romain Couillet, and Walid Hachem.
\newblock {Statistical inference in large antenna arrays under unknown noise
  pattern}.
\newblock {\em IEEE Transactions on Signal Processing}, 61(22):5633--5645,
  2013.

\end{thebibliography}

\end{document}